\documentclass[preprint,12pt]{report}

\usepackage[utf8]{inputenc}
\usepackage[T1]{fontenc}

\usepackage{amssymb}
\usepackage{amsmath,amssymb,bm}
\usepackage{xfrac}
\usepackage{graphicx}
\usepackage{hyperref}
\usepackage{verbatim}
\usepackage[linesnumbered,lined,boxed,commentsnumbered,ruled]{algorithm2e}
\SetKwComment{Comment}{}{}
\usepackage{algpseudocode}
\usepackage{listings}
\usepackage{changepage}
\usepackage{xspace}
\usepackage{adjustbox}
\usepackage{multirow}
\usepackage{adjustbox}
\usepackage{multirow}
\usepackage{amsmath}
\usepackage{amsfonts}
\usepackage{booktabs}
\usepackage{float}
\usepackage{pgfplots} 
\usepackage{comment}

\usepackage{lmodern}
\usepackage{anyfontsize}
\usepackage{figparameters}
\usepackage[normalem]{ulem}     
\usepackage{xcolor}             

\pgfplotsset{compat=1.16} 

\hypersetup{pdfauthor=author}

\usepackage{defs}
\usepackage{fca}

\begin{document}

\title{A minimal base or a direct base? That is the question!}


\author{Jaume Baixeries \\
Computer Science Department. \\ Universitat Politècnica de Catalunya \\
Jordi Girona, 1-3, Barcelona, Catalonia \\
Amedeo Napoli \\
Universit{\'e} de Lorraine, CNRS, LORIA \\
54000 Nancy \\ 
France \\
}

\maketitle

\begin{abstract}

In this paper we revisit the problem of computing the closure of a set of attributes given a basis of dependencies or implications.
This problem is of main interest in logics, in the relational database model, 
in lattice theory, and in Formal Concept Analysis as well.
A basis of dependencies may have different characteristics, 
among which being ``minimal'', e.g., the \dgbasis, or being ``direct'', 
e.g., the \canonical and the \dbasis.
Here we propose an extensive and experimental study of the impacts of minimality  and directness on the closure algorithms.
The results of the experiments performed on real and synthetic datasets are analyzed in depth, and suggest a different and fresh look at computing the closure of a set of attributes w.r.t. a basis of dependencies.

This paper has been submitted to the 
\textit{International Journal of Approximate Reasoning}.


\end{abstract}




%
\section{Introduction}
\label{sec:introduction}

In this paper, we are interested in analyzing the characteristics of different covers or bases of dependencies, the way they are used and their related 
efficiency.
A dependency or implication $\ad{X}{Y}$ can be read as \textit{$X$ implies $Y$} 
and follows the so-called Armstrong axioms \cite{Armstrong74}.
Dependencies are ``first class citizens'' in different fields of Computer Science, e.g., 
Horn clauses in logics, functional dependencies in the relational database model, 
implications in Formal Concept Analysis (FCA).

This paper is a follow-up of \cite{BaixeriesCKN23} where we studied three different covers, 
namely the \minimal in relational database theory \cite{Maier83}, the \canonical (\cdub) 
in lattice theory \cite{BertetM10}, 
also known as the \textit{canonical cover} in relational database theory 
(\cite{Maier83}, \cite{MannilaR92}),
and the \duquenne aka canonical basis in FCA \cite{GuiguesD86}.
These covers are introduced and characterized in many different textbooks,
e.g., in database theory \cite{Maier83,MannilaR92,AbiteboulHV95},
in logics \cite{CramaH11}, in lattice theory~\cite{BertetM10},
and in FCA \cite{GanterW99,GanterO16}.
Moreover, Marcel Wild proposes in \cite{Wild17} an extensive study about 
implication bases and the relations existing between the different fields in which they are used.

While the \cdub (canonical cover) is of first importance in database theory \cite{PapenbrockEMNRZ15}, 
the \duquenne (\dgbasis) is the implication basis of reference in FCA.
In particular, authors in \cite{BazhanovO11,BazhanovO14} analyze the computation 
of the \dgbasis w.r.t. three closure algorithms, namely \closure, \linclosure, and \wild.
In this paper we follow these tracks and we extend this seminal work in several directions, 
as we analyze not only the \dgbasis but also the \cdub and the \dbasis \cite{AdarichevaNR13,AdarichevaNV24}.
The \dbasis has gained more importance these last years and the comparison of these three implication bases is one originality of the present paper.
For doing so, we characterize the behaviors of several combinations of the three closure algorithms and we try to evaluate the importance for a basis of being minimal or direct.

The construction of a cover depends on computing the closure $\clo(X)$ of a set of attributes $X$ w.r.t. a set of dependencies $\Sigma$ thanks to the Armstrong axioms. 
Given a set of dependencies $\Sigma$, 
there may exist different sets of dependencies that are \textit{equivalent} modulo Armstrong axioms.
Then two alternative cases can be considered for covers,
(i) a cover is \textit{minimal} when it contains a minimal number of dependencies, 
i.e., minimal in order to maintain the equivalence modulo Armstrong axioms,
(ii) a cover is \textit{direct} if only one pass over the set $\Sigma$ 
is sufficient to compute the closure $\clo(X)$ for any set of attributes $X$.
For example, the \dgbasis is minimal while the \cdub and the \dbasis are direct.

Given a set of dependencies $\Sigma$ and a set of attributes $X$, to decide what should be the characteristics of $\Sigma$ to be used to perform the computation of $\clo(X)$ is an important problem because the number of dependencies that may hold in an even small dataset can be huge, and because costly operations are applied to  $\Sigma$.
Then the debate can be stated in the following terms:
(\textit{q1}) is it better to have a cover with a smaller set of dependencies that may require more than one pass to compute a closure,
or dually, (\textit{q2}) is it better to have a larger cover ensuring that only one pass is required to compute the closure?
To be complete, the question of the algorithm computing $\clo(X)$ should also be raised.

Checking whether it is better to use a direct basis 
or a minimal basis has not yet been fully explored.
For example, the minimality of an implication basis has an effective impact 
on a process such as attribute exploration and its application to knowledge engineering, 
see e.g, \cite{BaaderGSS07,BaaderD08,VolkerR08,RysselDB14}.
In addition, the fact that an implication basis is direct received a lot of attention 
in lattice theory~\cite{BertetM10,AdarichevaNR13,AdarichevaN17} and in
FCA~\cite{GanterW99,GanterO16,LorenzoBCEM18},
while this characteristic is ignored in database theory even if the \cdub is 
the implication basis of reference.
Accordingly, a third question is addressed and discussed in this paper:
(\textit{q3}) regardless of the hypothetical reasons why a direct basis is preferred in database theory instead of a minimal basis, 
can a \dgbasis be accepted outside the FCA community as a valid alternative to compute $\clo$?
In order to answer such questions we have performed a set of experiments over real and synthetic datasets,
comparing the computing of the closure $\clo$ with \closure, \linclosure, 
and \wild in combination with two direct bases, \cdub and \dbasis, and a minimal basis, the \dgbasis.
The results presented herein suggest that in certain conditions (detailed farther) 
the \dgbasis may be a robust and valid alternative to the \cdub and the \dbasis
to compute $\clo$.

This paper is organized as follows.
Firstly we introduce the definitions used in this paper in Section~\ref{sec:definitions}.
In Section~\ref{sec:algos} we propose three algorithms for computing a closure, 
namely \closure, \linclosure, and \wild.
We present the characteristics of the three bases of dependencies in Section~\ref{sec:three-bases}.
In Section~\ref{sec:impact} we analyze the impacts of using a direct basis when computing a closure and we propose new versions of the three algorithms adapted to the processing of direct bases.
Finally, we present a series of experiments in Section~\ref{sec:experiments} 
whose results are analyzed in Section~\ref{sec:results},
before ending with a discussion and conclusions in Section~\ref{sec:discussion}.



\section{Definitions}
\label{sec:definitions}

Although we refer to \cite{Date00} in most of the cases, the definitions below can be found as well in many different textbooks and papers related to database theory, logics, and FCA.
All along this paper, we consider a tabular dataset whose column labels form the \textit{set of attributes~$\cjtatrib$}, i.e., the set of interest in the following.
The row labels of the dataset determine the transactions or the objects whose descriptions are given by the columns.
We use capital letters to denote sets of attributes: $A,B,X,Y \dots \subseteq \cjtatrib$
and lowercase letters for single attributes: $a,b,x,y \dots \in \cjtatrib$.

Given $X, Y \subseteq \cjtatrib$, the fact that a dependency $\ad{X}{Y}$ is \textit{valid} or \textit{true} depends on the kind of dependency at hand.
For example, an instance in which a Horn clause is true is an assignment of propositional variables.
An instance in which a functional dependency is valid is a set of rows 
in a many-valued tabular dataset.
In the same way, an instance where an implication is true in a formal context 
is a given set of objects.

Then, the dependency $\ad{X}{Y}$ \textit{holds} should be understood as $\ad{X}{Y}$ holds for all the instances where it is valid or true.
In addition, ``If $\ad{X}{Y}$ holds, then $\ad{XZ}{YZ}$ holds''
can be rephrased as
``In any instance in which $\ad{X}{Y}$ is valid, the dependency $\ad{XZ}{YZ}$ is valid as well''.

\begin{definition}[\cite{Date00}]
\label{def:armstrong}
Given a set of attributes $\cjtatrib$, for any $X,Y,Z \subseteq \cjtatrib$, 
the Armstrong axioms are:
\begin{enumerate}
    \item \textbf{Reflexivity}: If $Y \subseteq X$, then $\ad{X}{Y}$ holds.
    \item \textbf{Augmentation}. If $\ad{X}{Y}$ holds, then $\ad{XZ}{YZ}$ holds.
    \item \textbf{Transitivity}. If $\ad{X}{Y}$ and $\ad{Y}{Z}$ hold, then $\ad{X}{Z}$ holds.
\end{enumerate}
\end{definition}

The Armstrong axioms allow us to define the closure of a set of dependencies as the iterative application of these axioms over a set of dependencies.

\begin{definition}[\cite{Date00}]
\label{def:closureDeps}
$\Sigma^+$ denotes the closure of a set of dependencies $\Sigma$
and is constructed by the iterative application of the Armstrong axioms over $\Sigma$.

This iterative application terminates when no new dependency can be added, and it is finite.
Therefore, $\Sigma^+$ contains the largest set of dependencies holding in all instances in which all the dependencies in $\Sigma$ hold.
\end{definition}

The closure of a set of dependencies induces the definition of the cover of such a set of dependencies.

\begin{definition}[\cite{Date00}]
\label{def:cover}
The \textbf{cover} or \textbf{basis} of a set of dependencies $\Sigma$ is any set $\Sigma'$ such that $\Sigma'^+ = \Sigma^+$.
\end{definition}

We define now the \textit{closure} of a set of attributes $X \subseteq \cjtatrib$ with respect to a set of dependencies $\Sigma$.

\begin{definition}[\cite{Date00}]
\label{def:closureattrib}
The \textbf{closure} of $X$ with respect to a set of dependencies $\Sigma$ is
$$\clo(X) = X \cup \conjunt{ Y \mid \ad{X}{Y} \in \Sigma^+}$$
\end{definition}

In other words, the closure operation w.r.t. $\Sigma$ returns the largest set of attributes
implied by $X$, denoted as $\Sigma \models \ad{X}{\clo(X)}$.
Therefore, the implication problem $\Sigma \models \ad{X}{Y}$ boils down to testing whether $Y \subseteq \clo(X)$ (see~\S~4 in \cite{BeeriB79}).


Now we introduce two main characteristics of a cover, being \textit{direct}
and being \textit{minimal}.
The definition of a minimal cover is independent of how $\clo(X)$ is computed:

\begin{definition}
\label{def:minimal}
Let $\Sigma$ be a set of dependencies.
We say that $\Sigma_{min}$ is a \textbf{minimal basis} of $\Sigma$ iff:
\begin{enumerate}
\item $\Sigma^+ = \Sigma_{min}^+$.
\item There is no smaller basis verifying the above property.
\end{enumerate}
\end{definition}

We now propose an alternative definition of the closure of a set of attributes, 
contrasting Definition~\ref{def:closureattrib}.

\begin{definition}[\cite{BertetM10}]
\label{def:pass}
Let $\Sigma$ be a set of dependencies and let $X \subseteq \cjtatrib$
be a set of attributes.
A \textbf{pass} over $X$ w.r.t. $\Sigma$ is defined as:
\[ \pass(X) = X \cup \conjunt{B \mid A \subseteq X \text{ and } \ad{A}{B} \in \Sigma} \]
\end{definition}

Then, the closure of a set of attributes $X$ can be defined as follows:

\begin{definition}[\cite{BertetM10}]
\label{def:direct}
Let $\Sigma$ be a set of dependencies. 
\[\clo(X) = \Pi^1_\Sigma(X) \cup \Pi^2_\Sigma(X) \cup \dots \cup \Pi^{k-1}_\Sigma(X)\]
where $\Pi^1_\Sigma(X) = \Pi_\Sigma(X)$ and $\Pi^i_\Sigma(X) = \pass(\Pi^{i-1}_\Sigma(X))$.
\end{definition}

Thus the computing of $\clo(X)$ relies first on computing $\pass(X)$, and then,
computing $\pass(\pass(X))$, and so on, until a fixed point
$\Pi^k_\Sigma(X) = \Pi^{k-1}_\Sigma(X)$ is reached.
We can now define a direct basis:

\begin{definition}[\cite{BertetM10}]
\label{def:directbase}
Let $\Sigma$ be a set of dependencies.
$\Sigma$ is a \textbf{direct basis} if for all $X \subseteq \cjtatrib$:
\[ \clo(X) = \pass(X) \]
\end{definition}

The notion of direct basis is well-known in lattice theory and FCA, 
but seems to be completely alien to the DB community.
References to a direct basis can be found in~\cite{BertetM10} and, earlier, in \cite{GanterW99}.



\def\algClosure{
\begin{function}
\caption{Closure($X,\Sigma$)}\label{alg:closure}
\SetKwInOut{Input}{Input}\SetKwInOut{Output}{Output}
\DontPrintSemicolon

\Input{A set of attributes $X \subseteq \cjtatrib$ and a set of implications
$\Sigma$}
\Output{$\clo(X)$}

\BlankLine

$stable \gets \textbf{false}$\;

\While(\tcp*[f]{Outer loop}){not $stable$} 
{
    $stable \gets \textbf{true}$\;
    \ForAll(\tcp*[f]{Inner loop}){$\ad{A}{B} \in \Sigma$} 
    { 
            \If(\tcp*[f]{deps}){$A \subseteq X$}
            {
                $X \gets X \cup B$\;
                $stable \gets false$\;
                $\Sigma \gets \Sigma \setminus \conjunt{\ad{A}{B}}$\;
            }
    }
}
\Return{X}\;

\end{function}
}

\def\algLinclosure{
\begin{function}
\caption{LinClosure($X,\Sigma$)}\label{alg:linclosure}
\SetKwInOut{Input}{Input}\SetKwInOut{Output}{Output}
\DontPrintSemicolon

\Input{A set of attributes $X \subseteq \cjtatrib$ and a set of implications $\Sigma$}
\Output{$\clo(X)$}

\BlankLine

\ForAll(\tcp*[f]{Preparation}){$\ad{A}{B} \in \Sigma$} 
{
    $count[\ad{A}{B}] \gets \mid A \mid$


    \ForAll{$a \in A$}
    {
        $list[a] \gets list[a] \cup \conjunt{\ad{A}{B}}$
    }
}

$update \gets X$

\BlankLine

\While(\tcp*[f]{Outer loop}){$update \neq \emptyset$} 
{
    choose $m \in update$
    
    $update \gets update \setminus \conjunt{m}$

    \BlankLine

    \ForAll(\tcp*[f]{Inner loop }){$\ad{A}{B} \in list[m]$} 
    { 
        $count[\ad{A}{B}] \gets count[\ad{A}{B}] - 1$
        

        \If(\tcp*[f]{deps}){$count[\ad{A}{B}] = 0$}
        {
                $add \gets B \setminus X$
                
                $X \gets X \cup add$
                
                $update \gets update \cup add$
        }
    }
}

\Return{X}\;

\end{function}
}

\def\algWildClosure{
\begin{function}
\caption{WildClosure($X,\Sigma$)}\label{alg:wildclosure}
\SetKwInOut{Input}{Input}\SetKwInOut{Output}{Output}
\DontPrintSemicolon

\Input{A set of attributes $X \subseteq \cjtatrib$ and a set of implications $\Sigma$}
\Output{$\clo(X)$}

\BlankLine

\ForAll(\tcp*[f]{Preparation}){$m \in \cjtatrib$} 
{
    \ForAll{$\ad{A}{B} \in \Sigma$}
    {
        \If{$m \in A$}
        {
            $list[m] = list[m] \cup \conjunt{\ad{A}{B}}$
        }
    }
}

\BlankLine

$stable \gets false$

\While(\tcp*[f]{Outer loop}){$not~stable$} 
{
    $stable \gets true$
    
    $\Sigma_1 \gets \bigcup_{m ~\in~ \cjtatrib \setminus X} list[m]$
    
    \BlankLine

    \ForAll(\tcp*[f]{Inner loop / deps}){$\ad{A}{B} \in \Sigma \setminus \Sigma_1$} 
    { 
            $X \gets X \cup B$
            
            $stable \gets false$
    }

    $\Sigma \gets \Sigma_1$
}
\Return{X}\;

\end{function}
}

\def\algClosureAccum{
\begin{function}[H]
\caption{Closure++($X,\Sigma$)}\label{alg:closure_accu}
\SetKwInOut{Input}{Input}\SetKwInOut{Output}{Output}
\DontPrintSemicolon

\Input{$X \subseteq \cjtatrib$ and a set of implications $\Sigma$}
\Output{$X \cup \conjunt{B \mid \ad{A}{B} \text{ and } A \subseteq X}$}

\BlankLine

$result \gets \emptyset$\;

\ForAll{$\ad{A}{B} \in \Sigma$} 
{ 
    \If{$A \subseteq X$}
    {
        $result \gets result \cup B$\;
    }
}

\Return{$X \cup result$}\;

\end{function}
\bigskip
}


\def\algClosureDirect{
\begin{function}
\caption{ClosureDirect($X,\Sigma$)}\label{alg:closuredirect}
\SetKwInOut{Input}{Input}\SetKwInOut{Output}{Output}
\DontPrintSemicolon

\Input{A set of attributes $X \subseteq \cjtatrib$ and a 
\textbf{direct} basis of implications $\Sigma$}
\Output{$\clo(X)$}

\BlankLine

\ForAll(\tcp*[f]{Inner loop}){$\ad{A}{B} \in \Sigma$} 
{ 
    \If(\tcp*[f]{deps}){$A \subseteq X$}
    {

            $X \gets X \cup B$\; 
        }
}

\Return{X}\;

\end{function}
}

\def\algLinclosureDirect{
\begin{function}
\caption{LinClosureDirect($X,\Sigma$)}\label{alg:linclosuredirect}
\SetKwInOut{Input}{Input}\SetKwInOut{Output}{Output}
\DontPrintSemicolon

\Input{$X \subseteq \cjtatrib$ and a 
\textbf{direct} basis of implications $\Sigma$}

\Output{$\clo(X)$}

\BlankLine

\tcp*[l]{The preparation part is the same as in Linclosure}

\If{$\Sigma$ \text{is a \dbasis}}
{
    $update \gets \Sigma_0(X)$
}
\Else
{
    $update \gets X$
}

$add \gets \emptyset$

\While(\tcp*[f]{Outer loop}){$update \neq \emptyset$} 
{
    choose $m \in update$
    
    $update \gets update \setminus \conjunt{m}$

    \BlankLine

    \ForAll(\tcp*[f]{Inner loop }){$\ad{A}{B} \in list[m]$} 
    { 
       $count[\ad{A}{B}] \gets count[\ad{A}{B}] - 1$
        
        \If(\tcp*[f]{deps}){$count[\ad{A}{B}] = 0$}
        {

                $add \gets add \cup B$
                
        }
    }
}

\Return{$X \cup add$}
\end{function}
}

\def\algWildClosureDirect{
\begin{function}
\caption{WildClosureDirect($X,\Sigma$)}\label{alg:wildclosuredirect}
\SetKwInOut{Input}{Input}\SetKwInOut{Output}{Output}
\DontPrintSemicolon

\Input{A set of attributes $X \subseteq \cjtatrib$ and a 
\textbf{direct} basis of implications $\Sigma$}\Output{$\clo(X)$}

\tcp*[l]{The preparation part is the same as WildClosure}

\If{$\Sigma$ \text{is a \dbasis}}
{
    $X \gets \Sigma_0(X)$
}
\Else
{
    $X \gets X$
}

\BlankLine

$\Sigma_1 \gets \bigcup_{m ~\in~ \cjtatrib \setminus X} list[m]$

\BlankLine

\ForAll(\tcp*[f]{Inner loop / deps}){$\ad{A}{B} \in \Sigma \setminus \Sigma_1$} 
{ 
        $X \gets X \cup B$
}

\Return{X}\;
\end{function}
}



\section{Algorithms Computing the Closure of a Set of Attributes}
\label{sec:algos}

In this section, we focus on the most well-known algorithms computing the closure of a set of attributes $X$, namely \closure, \linclosure, and Wild Closure.

\subsection{The \closure Algorithm}
\label{sec:closure}

\algClosure

\closure is the \textit{classical} algorithm computing $\clo(X)$,
which is detailed in many textbooks, e.g., in \cite{Maier83,AbiteboulHV95,GanterO16}.
Here, we adapt the version proposed  in~\cite{GanterO16} (Algorithm~14, page~93).
The computing of a given pass $\pass(X)$ over $X$ is performed in lines $4-10$.
The outer loop is iterated in line $2-11$ until a fixed point is found.
A dependency which is processed in lines $5-9$ is then removed in line $8$.

The complexity of this algorithm is discussed in the related references, 
and the general consensus is that it is \textit{quadratic} w.r.t. the input
(see \cite{BaixeriesCKN23} for more details).

\subsection{The \linclosure Algorithm}
\label{sec:linclosure}

\algLinclosure

The \linclosure algorithm \cite{BeeriB79} is an improved version of \closure consisting
in two parts:
the \texttt{preparation} in which the necessary data structures are  computed,
and the \texttt{computation} in which $\clo(X)$ is computed.
In \texttt{preparation}, two data structures are constructed, the role of which
is to ensure that only dependencies necessary to compute the closure
are considered while the others are ignored:

(i) for each attribute $x \in X$, the first structure records a pointer
to all dependencies $\ad{A}{B}$ such that $x$ appears in the left-hand side (LHS) $A$,

(ii) for each dependency $\ad{A}{B}$, the second structure includes a counter
recording the number of attributes in $A$ already visited during
computing.

Since only reduced and clarified contexts are processed in this paper (see Section~\ref{sec:experiments}), there is no dependency with an empty LHS.
This means that dependencies $\ad{A}{B}$ such that $|A| = 0$ are not considered in the preparation part of \ref{alg:linclosure}.

There is a general consensus about the complexity of \linclosure, 
which is linear w.r.t. the size of $\Sigma$ for both the preparation part
and the computation part~\cite{BeeriB79}.
The complexity of the preparation part is assumed to be of the same complexity as the rest of the algorithm.
One explanation of this fact appears in the pioneering paper~\cite{BeeriB79}, 
page $47$ in the second paragraph 
(this paragraph is adapted to fit names in Algorithm \linclosure):
\begin{quote}
\textit{For each attribute in [update], the [outer] loop follows a
  constant number of steps for each occurrence of that attribute on
  the left side of an FD in $\Sigma$.
  Similarly, each right side of an FD in $\Sigma$ is visited at most once
  in [the outer loop].
  Thus [the outer loop] is also $\bigO(|\Sigma|)$ as is the entire
  Algorithm.}
\end{quote}

\subsection{The \wild Algorithm}                                          %
\label{sec:wildclosure}                                                   %

Below we present a slightly more compact form of the \wild algorithm
borrowed from \cite{BazhanovO14}.
The \wild Algorithm \cite{Wild95} aims at ensuring that inside each outer loop
all the dependencies $\ad{A}{B}$ fulfilling the condition $A \subseteq X$
are selected.
The algorithm starts with one of the data structures also present
in \linclosure: for each attribute $x$ there is a list recording all the
dependencies $\ad{A}{B}$ such that $x$ is contained in $A$.
Then, it selects all dependencies $\ad{A}{B}$ such that $A \subseteq X$
to be processed in the loop in lines $12 - 15$.

Regarding the complexity of the algorithm, it is underlined in~\cite{Wild95} that:
\begin{quote}
\textit{
Algorithm 1 [Wild Closure] has complexity $\bigO(|\Sigma| |\cjtatrib|^2)$, 
which is actually the same as the complexity of Algorithm 0 [Closure].
Yet in practice Algorithm 1 [Wild Closure] takes a fraction of the time
of Algorithm 0 [Closure] and also of \linclosure.
Philosophy: Doing few set operations with big sets is better than doing
many set operations with small sets.}
\end{quote}

This apparent paradox between the asymptotic complexity of an algorithm
and its real performance is of interest and will be more deeply discussed
in Section~\ref{sec:experiments}.

\algWildClosure



\section{Three Bases of Dependencies}
\label{sec:three-bases}

In this section we briefly present the \cdub, the \dbasis, and the \dgbasis, 
that will be processed by the three algorithms introduced 
in Section~\ref{sec:algos}.
The \cdub and the \dbasis are direct and unit bases while the \dgbasis is a minimal non-unit basis.
Recall that the right-hand side (RHS) of a unit dependency is reduced to a single element and that a unit basis only includes units dependencies.

For the seek of fairness in the processing of all three bases, 
dependencies having the same LHS in the \cdub and the \dbasis
are replaced by a single dependency whose RHS is the union of the previous RHS.
For example, the unit dependencies $\ad{ab}{c}, \ad{ab}{d}, \ad{ab}{e}$ are replaced by $\ad{ab}{cde}$.

\subsection{The \canonical (\cdub)} 
\label{subsec:canonical}

The \cdub is deeply studied in~\cite{BertetM10} where different equivalent variations are examined.
This basis can be characterized as follows:

\begin{enumerate}
\item
  $\Sigma$ is a unit basis, i.e., the RHS of all dependencies in $\Sigma$ only include one single attribute.
\item
  $\Sigma$ is left-reduced.

  A dependency $\ad{X}{y}$ is left-reduced if, for all $X_i \subseteq X$,
  the dependencies $\ad{X_i}{y}$ do not hold.
  Stated differently, in a \cdub a LHS is minimal and is also called
  a \textit{minimal generator}.
\end{enumerate}

The \cdub may contain some redundancy.
For example, while the basis
$\Sigma = \conjunt{\ad{a}{b},\ad{b}{c},\ad{a}{c}}$ is left-reduced,
the dependency $\ad{a}{c}$ is redundant because
$\Sigma^+ = (\Sigma \setminus \conjunt{\ad{a}{c}})^+$.
The \cdub is not necessarily minimal but it is direct as specified 
in~Definition \ref{def:directbase} \cite{BertetM10}.


%
\subsection{The \dbasis}
\label{subsec:dbasis}

The \dbasis is introduced in \cite{AdarichevaNR13} and revisited in \cite{AdarichevaNV24} 
as a subset of the \cdub, i.e., it can be constructed by removing some dependencies from a \cdub.
In~\cite{AdarichevaNV24} it is explained that 
\textit{there may be an exponential gap between the two bases}.
Here we will follow the definition given in \cite{AdarichevaNV24} while an 
alternative definition is proposed in~\cite{LorenzoACEM15}.

The \dbasis can be considered as a subset of the \cdub composed of two subsets of implications.
The first subset, denoted as $\Sigma_0$, is a set of binary implications, 
i.e., LHS and RHS are of size $1$, and includes all binary implications lying in the \cdub.
$\Sigma_0 = \{\ad{a}{c}, c \in \phi(a) \setminus \{a\}, a,c \in \atrib\}$,
where $\phi$ is the closure operator related to the \dbasis
\footnote{%
      As underlined in~\cite{AdarichevaNV24} the closure operator $\phi$ is supposed to be \textit{standard}, 
      i.e., $\forall a \in \atrib, \phi(a) \setminus \{a\}$ is closed 
      and $\phi(a) \neq \phi(b)$ unless $a=b$, 
      which means in terms of FCA that the related context is clarified and reduced.
      Moreover, still in~\cite{AdarichevaNV24}, $\Sigma_0$ is denoted as $\Sigma^b$ and $\clozero$ as $\phi^b$.
}.
The second subset is defined w.r.t. the operator $\clozero$ and the notion of \textit{minimal D-generator}.
The operator $\clozero(X)$ computes the closure of a set $X$ w.r.t. 
the set of implications $\Sigma_0$.

\begin{definition}
\label{def:dbase_dgen}
A minimal generator $A$ of $c$, with $|A| \geq 2$, is a \textit{minimal D-generator} of $c$ if:
\begin{itemize}
\item
  $c \not\in \clozero(A)$,
\item
  for every minimal generator $A'$ of $c$, $A' \subseteq \clozero(A)$ implies $A' = A$.
\end{itemize}
\end{definition}

Then, $\Sigma_D = \Sigma_0 \cup \{\ad{A}{c}, A \in gen_D(c), c \in \atrib\}$ 
defines the \dbasis over $\atrib$, where $gen_D(c)$ denotes the family of D-generators of $c$.
Finally, let us mention this quotation taken from~\cite{AdarichevaNR13}:  

\textit{%
While the D-basis is not direct in this meaning of this term
[this refers to Definition~\ref{def:directbase}],
the closures can still be computed in a single iteration of the basis,
provided the basis was put in a specific order prior to computation.
}

In particular, this is why the \dbasis is called an \textit{ordered direct implication basis}, 
contrasting the \cdub, where the order is irrelevant (see for example~\cite{Wild17}).


%
\subsection{The \dgbasis}
\label{subsec:dgbasis}

The \dgbasis~\cite{GuiguesD86,GanterW99}, also called the \textit{Canonical Basis} in the FCA community, relies on pseudo-closed sets~\cite{GanterW99,GanterO16}.
It should be noticed that the \dgbasis is also introduced in~\cite{Maier83} where it is called the \textit{\minimum basis}.
Below we first recall the definition of a pseudo-closed set of attributes and then 
the definition of the \dgbasis.

\begin{definition}
Let $\Sigma$ be a set of dependencies and $\cjtatrib$ the related set of attributes.
$X \subseteq \cjtatrib$ is \textbf{pseudo-closed} if:
\begin{enumerate}
\item
$X \neq \clo(X)$, i.e., $X$ is not closed.
\item
  If $Y$ is a proper subset of $X$ ($Y \subset X$) and is pseudo-closed,
  then $\clo(Y) \subseteq X$.
\end{enumerate}
\end{definition}

\begin{definition}
The \textbf{\dgbasis} of a set of dependencies $\Sigma$ is defined as:

\[\conjunt{\ad{X}{\clo(X)} \mid X \subseteq \cjtatrib \text{ and } X \text{ pseudo-closed}}\]
\end{definition}

The \dgbasis is not direct but it is \textit{minimal} and \textit{non-redundant}.



\section{Impact of a Direct Basis on the Three Algorithms}
\label{sec:impact}

In this section we discuss the impact of ``directness'' on the three algorithms
presented in Section~\ref{sec:algos}.
By \textit{impact} we mean the possibility of \textit{adapting} 
these algorithms to process a direct basis.
We explain, for each algorithm which changes can be performed when $\Sigma$ is a \cdub or a \dbasis.
We also try to quantify the potential benefits resulting from these changes.

We prove the correctness of the adapted algorithms by showing that they all compute
$\pass(X) = X \cup \conjunt{B \mid A \subseteq X \text{ and } \ad{A}{B} \in \Sigma}$ (Def. \ref{def:pass}).
Actually, we are using the definition of $\pass(X)$ as an algorithm that
(1) accumulates the RHS of all dependencies $\ad{A}{B}$ such that
$A \subseteq X$ and,
(2) joins this accumulation to $X$.
As discussed below, the order in which each step is \textit{separately} performed is relevant.

\subsection{The \dbasis and its Unit Implications}

Since the \cdub follows Definition \ref{def:directbase} (\cite{BertetM10}),
any algorithm that computes $\pass(X)$ with a \cdub
correctly computes $\clo(X)$ for all $X \subseteq \cjtatrib$.
However, when the input is a \dbasis, this is not necessarily true, as shown in the next example:

\begin{example}
\label{example:linclosure}
Let us consider the following reduced and clarified formal context:

\begin{center}
\begin{cxt}%
\cxtName{\context}%
\att{a}%
\att{b}%
\att{c}%
\att{d}%
\obj{x.x.}{$o_1$}
\obj{..xx}{$o_2$}
\obj{x...}{$o_3$}
\obj{.x..}{$o_4$}
\end{cxt}
\end{center}

The \dbasis related to this context is
$\Sigma = \conjunt{\ad{d}{c},\ad{bc}{a},\ad{ad}{b},\ad{ab}{c},\ad{bc}{d},\ad{ab}{d}}$.
The computation of $\pass(bd)$ iterates over all dependencies in $\Sigma$, and
the only dependency $\ad{A}{B}$ for which the test $A \subseteq bd$ is positive is $\ad{d}{c}$.
Then attribute $c$ is added to the result, and the returned value is $bcd$, 
which is not the correct answer for $\clo(bd)$.
\end{example}

It may be strange that $\clo(X) = \pass(X)$ does \textit{not} hold for a \dbasis, although it is also a direct basis.
As explained in Section \ref{subsec:dbasis}, the \dbasis is direct when 
the order in the processing of the dependencies is taken into account, 
i.e., dependencies whose LHS are singletons need to be processed first.
However, there is another relevant fact: the \dbasis needs to accumulate
the RHS of the processed dependencies to the closure being computed.

Indeed, in line $6$ of the \closure algorithm, i.e.,
$X \leftarrow X \cup B$,
the RHS of the dependency $\ad{A}{B}$ is accumulated \textit{inside the loop}
into variable $X$ that contains the partial result of $\clo(X)$.
By contrast, the definition of $\pass(X)$ accumulates the final result \textit{outside the loop}, once the loop has been exhausted.
This slight difference is crucial when analyzing the behavior
of the \dbasis when using $\pass(X)$ in order to compute $\clo(X)$.
Thus the \dbasis is direct also w.r.t. to an algorithm that accumulates the closure 
inside the loop.

This problem with the \dbasis is solved by computing $\pass(\clozero(X))$ instead of $\pass(X)$, 
where $\clozero(X)$ is the closure of $X$ w.r.t. unit dependencies,
i.e., whose LHS contains one single attribute (as discussed in Section~\ref{subsec:dbasis}).

\begin{proposition}
\label{prop:clopp}
If $\Sigma$ is a \dbasis, then $\clo(X) = \pass(\clozero(X))$.
\end{proposition}

\begin{proof}
We prove it by contradiction.
Suppose that we compute $\clo(X)$ as
$\pass(\clozero(X)) = \clozero(X) \cup \conjunt{B \mid A \subseteq \clozero(X) 
\text{ and } \ad{A}{B} \in \Sigma}$, 
but that the result is not correct because there are some attributes 
that are missing in the result, and let $c$ be one of these attributes.

We already know that $\clo(X) = \pass(X)$ when $\Sigma$ is a \cdub 
and that a \dbasis can be defined as a subset of a \cdub.
Let $\Sigma'$ be the equivalent \cdub to the \dbasis $\Sigma$.
If $\pass(X)$ fails to add the attribute $c$ to $\clo(X)$ 
this is because there is an dependency $\ad{A}{c} \in \Sigma'$ such that $A \subseteq X$ 
which is not in $\Sigma$.
This may be due to (see Definition~\ref{def:dbase_dgen}):

\begin{enumerate}
\item
  $c \in \clozero(A)$.
  But since we are computing the closure of $\pass(\clozero(X))$, 
  then $c \in \clozero(X)$ and attribute $c$ will appear in the result computed by $\pass(X)$.
  
\item
  There is a dependency $\ad{A'}{c}$ in $\Sigma$ such that
  $A' \subseteq \clozero(A)$ and $A' \neq A$.
  As we are computing the closure of $\clozero(X)$,
  $A' \subseteq \clozero(A) \subseteq \clozero(X)$, and
  then $\pass(X)$ will use dependency $\ad{A'}{c}$ to add $c$ to the result.
\end{enumerate}

\end{proof}

Thus, $\clo(X) = \pass(X)$ when the input is a \cdub,
and we have as well that $\clo(X) = \pass(\clozero(X))$ when the input is a \dbasis.
Moreover, it should be noticed that the three algorithms presented in Section~\ref{sec:algos} 
accumulate the result, i.e., $X \leftarrow X \cup B$, \textit{inside the loop}.
However, the modified algorithms \linclosuredir and \wilddir presented below will not.
This explains why \linclosuredir and \wilddir will need to compute first $\clozero(X)$
when $\Sigma$ is a \dbasis.

\subsection{Impact on \closure}
\label{sec:impactclosure}

When the input is a direct basis, the primary change to be performed 
in \closure is to remove the outer loop, since a direct basis ensures that only one loop is necessary.
As a consequence, it is no longer necessary to remove a dependency once it is processed (line 8) and the variable \textit{stable} is no longer needed.
A new version of \closure is given in Algorithm~\ref{alg:closuredirect}.

\algClosureDirect

It is straightforward to check that \closuredir computes $\clo(X)$ for any direct basis $\Sigma$:
the invariant of the loop in lines $1-5$ is 
$X = X \cup \conjunt{ B \mid \ad{A}{B} \text{ and } A \subseteq X }$ 
for all $\ad{A}{B}$ visited so far.
This means that $X$ will accumulate the right-hand sides of all those dependencies $\ad{A}{B}$
such that $A \subseteq X$, as required by Definition \ref{def:pass}.
At the end of this loop, we have that $X := \clo(X)$.

On the other hand, Section~11 in~\cite{AdarichevaNR13} states that 
\closuredir~\footnote{In~\cite{AdarichevaNR13} \closuredir is called \textit{Algorithm 0}.} 
correctly computes $\clo(X)$ when $\Sigma$ is a \dbasis.
In this case, we can remark that the input of \closuredir is $X$ and not
$\clozero(X)$, as the accumulation $X \leftarrow X \cup B$ is performed inside the loop.

\subsection{Impact on \linclosure}
\label{sec:impactlinclosure}

Compared to \closure, the outer loop of \linclosure does not scan per dependency but per attribute:
once the LHS of a dependency is checked as a subset of $X$, then the RHS is added to $\clo(X)$.
This means that performing just one single outer pass may not yield 
the correct computation of $\clo(X)$.
Consequently, the outer loop cannot be removed as in \closuredir, 
and the only way to reduce the number of loops and 
to guarantee the correct computation of $\clo(X)$ 
is to prevent the increase of the variable $update$.
This variable controls the number of inner loops:
for each attribute where $update$ is increased, \linclosure performs an extra inner loop.
Preventing $update$ from increasing can be achieved by removing lines 
$update \leftarrow update \cup add$ and $add \leftarrow B \setminus X$ 
in \linclosure, as shown in Algorithm~\ref{alg:linclosuredirect}.

\algLinclosureDirect

These changes return the right result when $\Sigma$ is a \cdub,
and when $\Sigma$ is a \dbasis provided that $\clo(\clozero(X))$ is computed
instead of $\clo(X)$, as proven in the next proposition.

\begin{proposition}
  \label{prop:linclosuredir}
  Given $X \subseteq \atrib$,
  if $\Sigma$ is a \cdub then \linclosuredir computes $\clo(X)$.
  If $\Sigma$ is a \dbasis then \linclosuredir with input $\clozero(X)$
  computes $\clo(X)$.
\end{proposition}

\begin{proof}
We show that after the outer loop in lines $9-18$, 
the variable $add$ in \linclosuredir contains 
$\conjunt{B \mid \ad{A}{B} \text{ and } A \subseteq X}$.
Let us consider $\ad{A}{B}$ in $\Sigma$ such that $A \subseteq X$.
At the beginning of the outer loop, $update$ contains all attributes in $X$ (line $6$).
In lines $10$ and $11$, an attribute is picked in $update$ (i.e., $X$) and then
removed from $update$.
The outer loop in lines $9-18$ ensures that all attributes in $update$ ($X$) 
are processed one by one at each loop.
In the inner loop, lines $12-17$, \linclosuredir marks all dependencies whose
LHS contains at least one attribute in $update$ (line $13$).
Since all attributes in $X$ are processed in the outer loop and since $A \subseteq X$, 
then $count[\ad{A}{B}]$ goes necessarily down to $0$.
Therefore line $15$ is executed, i.e., the RHS of $\ad{A}{B}$ is added to $\clo(X)$.

At the end, $update$ contains
$\conjunt{B \mid A \subseteq X \text{ and } \ad{A}{B} \in \Sigma}$,
and in line $19$ \linclosuredir returns
$X \cup \bigcup \conjunt{B \mid A \subseteq X \text{ and } \ad{A}{B} \in \Sigma}$.
\end{proof}

\subsection{Impact on \wild}
\label{sec:impactwild}

Contrasting \linclosure, the \wild algorithm checks that, at each pass of the
outer loop,  all dependencies $\ad{A}{B}$ such that $A \subseteq X$ are \textit{directly} processed.
This means that the containment test line $5$ in \closure and the test line $13$ 
over $count$ in \linclosure are no more necessary.

As previously, it should be ensured that \wild only performs one single pass of the outer loop.
Actually, the outer loop in \wild is equivalent to the outer loop in \closure, 
making things easier.
Thus, in \wilddir, it is only needed to remove the outer loop and consequently the 
variable $stable$ is no longer needed.

\algWildClosureDirect

\begin{proposition}
\label{prop:wilddir}
Given $X \subseteq \atrib$,
if $\Sigma$ is a \cdub then \wilddir computes $\clo(X)$.
If $\Sigma$ is a \dbasis then \wilddir with input $\clozero(X)$
computes $\clo(\clozero(X))$.
\end{proposition}

\begin{proof}
The key line of \wilddir algorithm is line $8$, where are selected all dependencies 
whose LHS contains an attribute not present in $X$.
Actually, if $a \in A$ in $\ad{A}{B}$ is such that $a \notin X$, 
then it is impossible that $A \subseteq X$.
Therefore, line $8$ of \wilddir ensures that all the dependencies used in the 
inner loop in lines $9-11$ are such that
$\conjunt{B \mid A \subseteq X \text{ and } \ad{A}{B} \in \Sigma}$.
Consequently, the RHS of these dependencies are added to $X$ in line $10$
and thus \wilddir algorithm computes 
$X \cup \bigcup \conjunt{B \mid A \subseteq X \text{ and } \ad{A}{B} \in \Sigma}$.
\end{proof}

We would like to mention that although \wilddir \textit{accumulates} the RHS in $X$, 
this does \textit{not} imply that its behavior is equivalent to \closuredir. 
The crucial fact here is that the set of dependencies over which 
the inner loop iterates \textit{does not change}.
For example, suppose that in line $10$ we increase $X$ with an arbitrary attribute 
$X := X \cup \conjunt{a}$ and that there is a dependency $\ad{X \cup \conjunt{a}}{B} \in \Sigma$.
The latter cannot be a candidate in the loop in lines $9-11$ because it contains $a \notin X$.
Now, one could think that because $a$ is added to $X$ 
this would mean that the dependency could be eventually selected.
In fact, in \closuredir it could be the case, while not in \wilddir, 
because the set of dependencies over which this algorithm iterates is already chosen 
in line $8$ and is not modified anymore.
This is why \wilddir performs the computation of $\clo(X)$ in the same way as
$\pass(X)$.
Stated differently, the RHS could be accumulated in a variable $result$ in line 10 and $X \cup result$ returned in line 12, mimicking Definition \ref{def:pass}.

\subsection{Complexity}
\label{sec:impactcomplexity}

We evaluate the actual improvement achieved  by \ref{alg:closuredirect}
w.r.t. \ref{alg:closure} when the input is a direct base such as a \cdub or a 
\dbasis.
In such a case, \closure computes $\clo(X)$ in the first loop,
and then, it loops over all the remaining dependencies in $\Sigma$,
i.e., those which are not removed in line $8$.
Since $\Sigma$ is direct, in the second loop $A \subseteq X$ is always tested
\textit{false}, and at the end of the second inner loop the function finishes.
Therefore, in the worst case scenario, \ref{alg:closuredirect}
saves a second loop of cost $\bigO(|\Sigma|)$.

In \linclosuredir, the number of iterations is reduced by preventing the variable 
$update$ to be increased.
As stated in Proposition~\ref{prop:linclosuredir}, 
the dependencies $\ad{A}{B}$ to be considered are such that $A \subseteq X$.
Regarding the remaining dependencies to be processed, the complexity of
\linclosure is of order $\bigO(|\Sigma|)$ (see~\S~\ref{sec:linclosure}), 
i.e., in the worst case, the whole set of dependencies could be iterated.

Comparing \linclosure and \linclosuredir, it should be noticed that in 
\linclosuredir $update$ and $X$ are never increased by the addition of new 
attributes, contrasting the case of \linclosure and leading to two collateral 
effects:

\begin{enumerate}
\item
  The fact that the variable $update$ (that is initialized with $X$)
  is never increased implies that the test in line $14$
  of \linclosuredir is less likely to test positive than the equivalent
  test in line $13$ in \linclosure, because in the
  latter case the variable $X$ may grow.
  This means that the number of processed dependencies as well as the
  number of attribute operations for the same closure may be higher
  in \linclosure than in \linclosuredir.
\item
  Since the variable $update$ is only decreased, it means that the number
  of loops is likely to be less in \linclosuredir than in \linclosure.
\end{enumerate}

We only have the guarantee that the number of processed dependencies and iterations
in \linclosuredir is less than in \linclosure, 
but we can only experimentally quantify the collateral effect mentioned above.

Comparing \wild and \wilddir, we can follow the discussion about \linclosure.
The complexity of \wild (\S~\ref{sec:wildclosure}) is of order $\bigO(\Sigma)$, 
i.e., in the worst case it may iterate all the dependencies in $\Sigma$.
By contrast, \wilddir will only iterate over the dependencies $\ad{A}{B}$ such that $A \subseteq X$.





\def\compCANONICALreal{ 
\begin{figure} 
	\centering
	\includegraphics[width=1.0\textwidth]{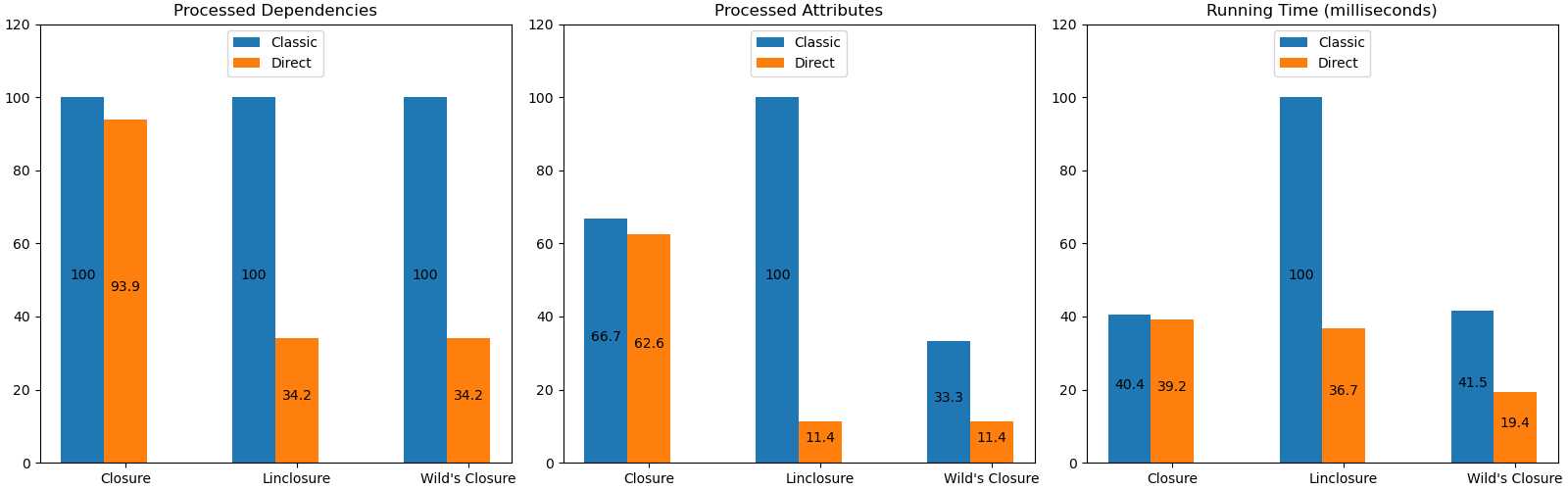}
	\caption{%
    Comparison of the performance of each algorithm w.r.t. their direct versions when processing the \cdub in \textbf{\real datasets}.
    The values have been normalized to the interval (0,100).}
	\label{fig:comparacio_real_CanonicalDirectUnitBasis}
\end{figure}
} 

\def\compDBASISreal{ 
\begin{figure} 
	\centering
	\includegraphics[width=1.0\textwidth]{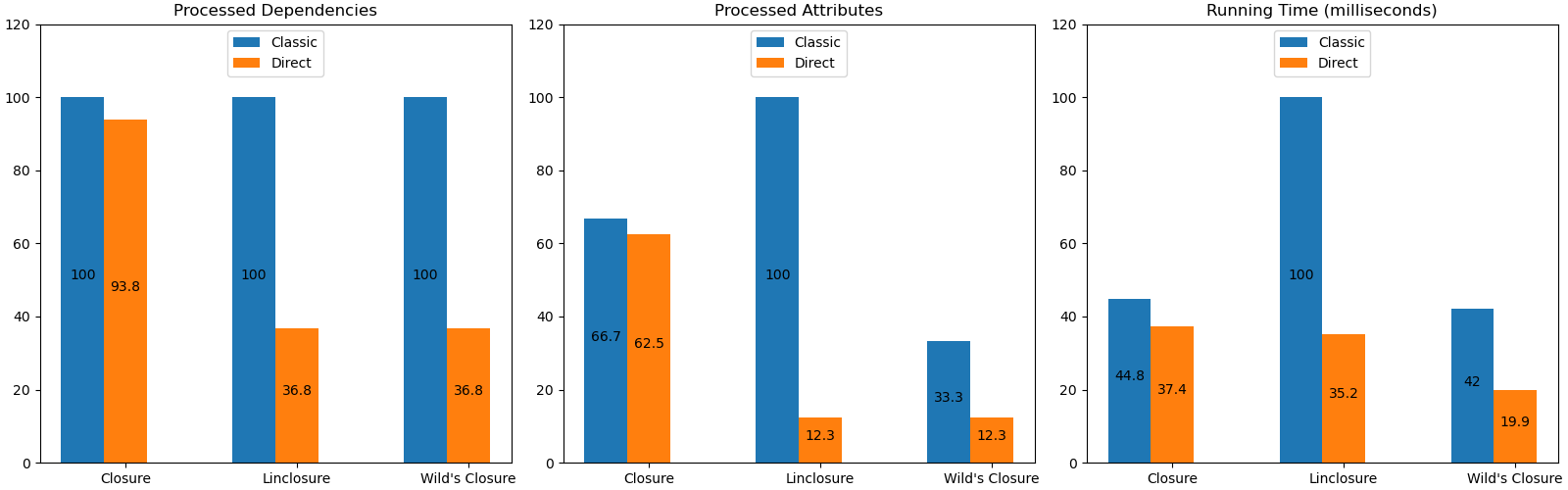}
	\caption{%
    Comparison of the performance of each algorithm w.r.t. their direct versions when processing the \dbasis in \textbf{\real datasets}.
    The values have been normalized to the interval (0,100).}
	\label{fig:comparacio_real_DBasis}
\end{figure}
} 

\def\barresreal{ 
\begin{figure} 
	\centering
	\includegraphics[width=1.0\textwidth]{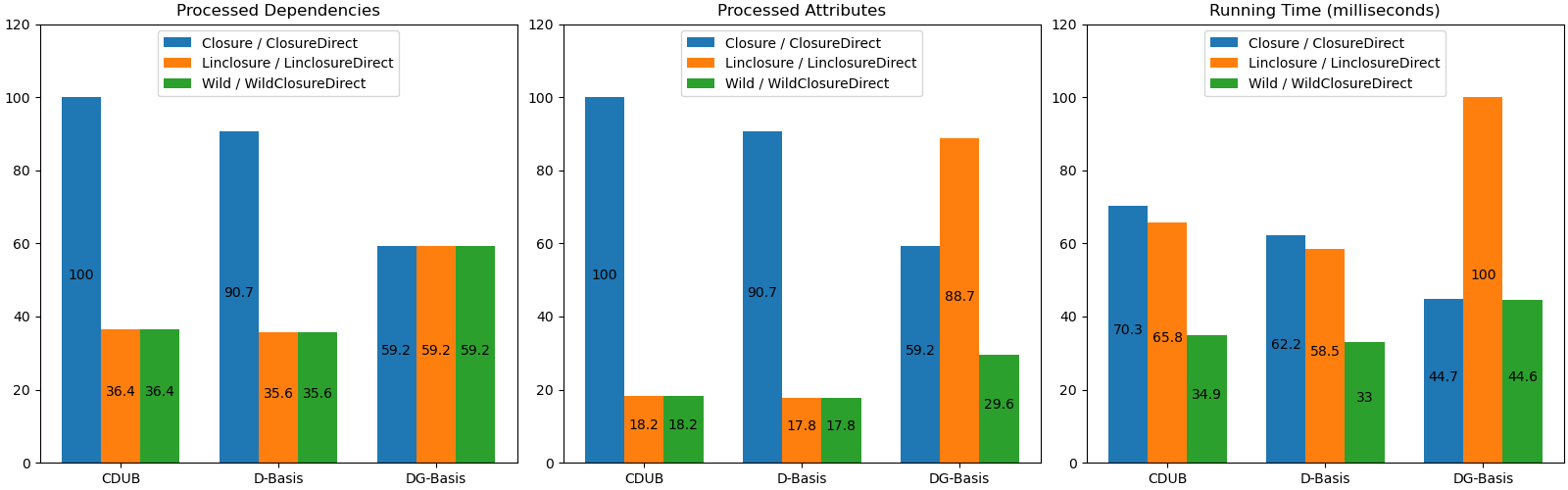}
	\caption{%
    Totals for the analyzed measures for each combination (Base $\times$ Algorithm) in \textbf{\real datasets}.
    The \cdub and the \dbasis are processed with the \textit{direct} versions
    of the algorithms. 
The values have been normalized to the interval (0,100).}
	\label{fig:barres_realdeps-atrib-in-out}
\end{figure}
} 

\def\barresNoDirectreal{ 
\begin{figure} 
	\centering
	\includegraphics[width=1.0\textwidth]{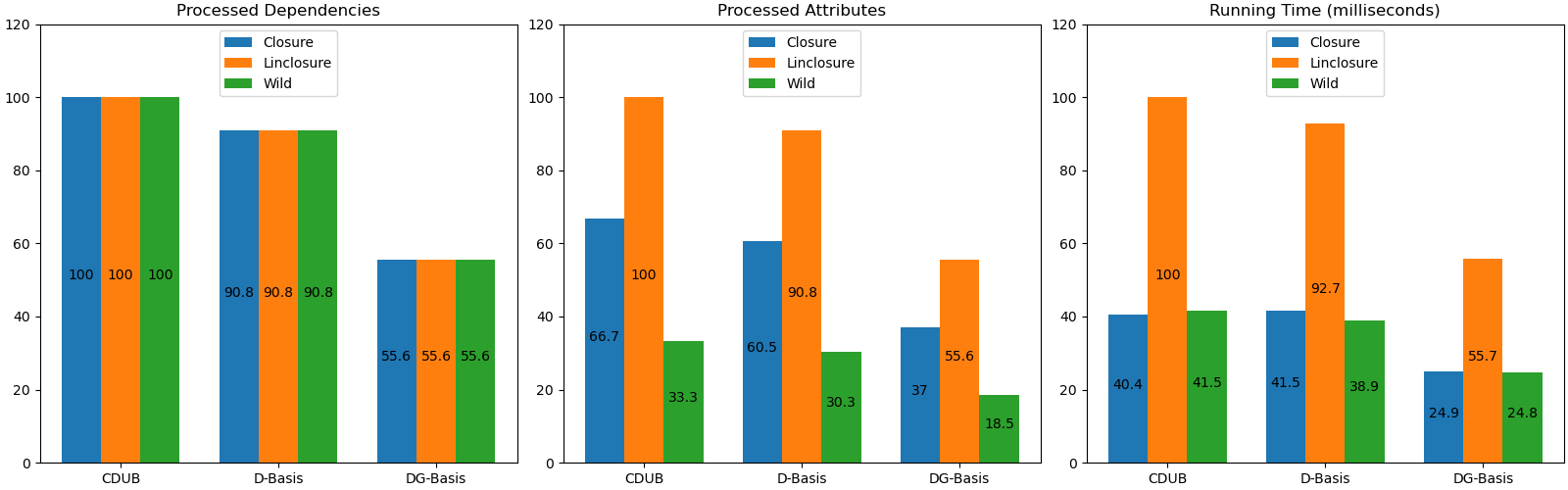}
	\caption{%
    Totals for the analyzed measures for each combination (Base $\times$ Algorithm) in \textbf{\real datasets}.
    All bases are processed with the \textit{classical} versions of the algorithms. 
    The values have been normalized to the interval (0,100).}
	\label{fig:barres_realdeps-atrib-no-direct}
\end{figure}
} 

\def\liniesrealatribs{ 
\begin{figure} 
	\centering
	\includegraphics[width=1.0\textwidth]{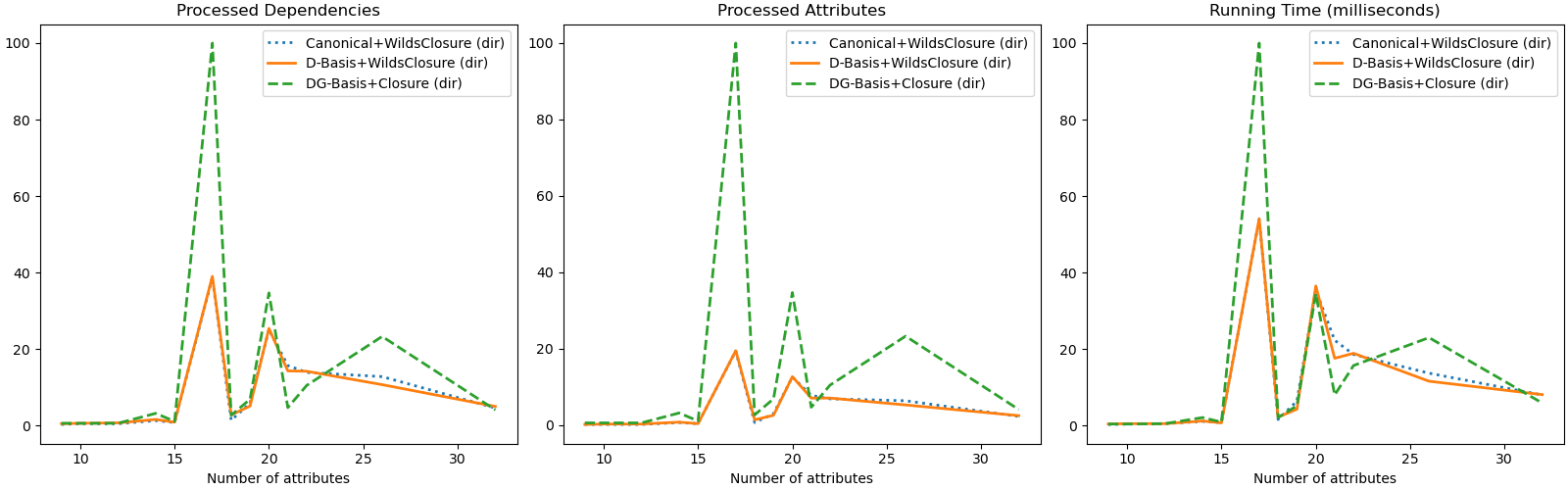}
	\caption{%
    Performance of the best combinations of (Base $\times$ Algorithm) for the analyzed metrics w.r.t. the number of attributes in \textbf{\real datasets}. 
The values have been normalized to the interval (0,100).}
	\label{fig:linies_real_liniesrealatribs}
\end{figure}
} 

\def\taulareal{ 
\begin{table} 
\centering 
\resizebox{0.9\textwidth}{!}{ 
\begin{tabular}{|l||c|c|c||c|c|c||c|c|c|} 
\hline 
 & \multicolumn{3}{|c|}{\cdub} & \multicolumn{3}{|c|}{\dbasis} & \multicolumn{3}{|c|}{DG-basis} \\\hline 
 Attribute &  CLO & LIN & WILD & CLO & LIN & WILD & CLO & LIN & WILD \\ 
\hline 
\hline 
Processed Dependencies & 0 & 10 & 10 & 0 & 4 & 4 & 5 & 5 & 5  \\ 
\hline 
Processed Attributes & 0 & 10 & 10 & 0 & 4 & 4 & 0 & 0 & 5  \\ 
\hline 
Running Time (milliseconds) & 0 & 0 & 7 & 1 & 0 & 5 & 6 & 0 & 0  \\ 
\hline 
\hline 
\end{tabular} 
} 
\caption{%
Best performance in the \textbf{\real datasets} for each pair base+algorithm for the 5 metrics.
The total number of databases is 19 (for each metric there can be more than one minimal combination).} 
\label{table:ranking_real} 
\end{table} 
} 

\def\comparacioCDUBDGBasisreal{ 
\begin{figure} 
	\centering
	\includegraphics[width=0.7\textwidth]{comparacio_CDUB_DGBasis_real.png}
	\caption{%
		X-axis: proportion $\sfrac{|\cdub|}{|DGBasis|}$. 
		Y-axis: proportion of combinations (Base $\times$ Algorithm) 
		that have the best execution time in \textbf{\real datasets}. 
		All values have been normalized to the interval (0,100).
	}
	\label{fig:comparacioCDUBDGBasisreal}
\end{figure}
} 

\def\comparacioCDUBDBasereal{ 
\begin{figure} 
	\centering
	\includegraphics[width=0.7\textwidth]{comparacio_CDUB_DBase_real.png}
	\caption{%
		X-axis: proportion $\sfrac{|\cdub|}{|\dbasis|}$. 
		Y-axis: proportion of combinations (Base $\times$ Algorithm) 
		that have the best execution time in \textbf{\real datasets}. 
		All values have been normalized to the interval (0,100).
	}
	\label{fig:comparacioCDUBDBasereal}
\end{figure}
} 





\def\compCANONICALsynthetic{ 
\begin{figure} 
	\centering
	\includegraphics[width=1.0\textwidth]{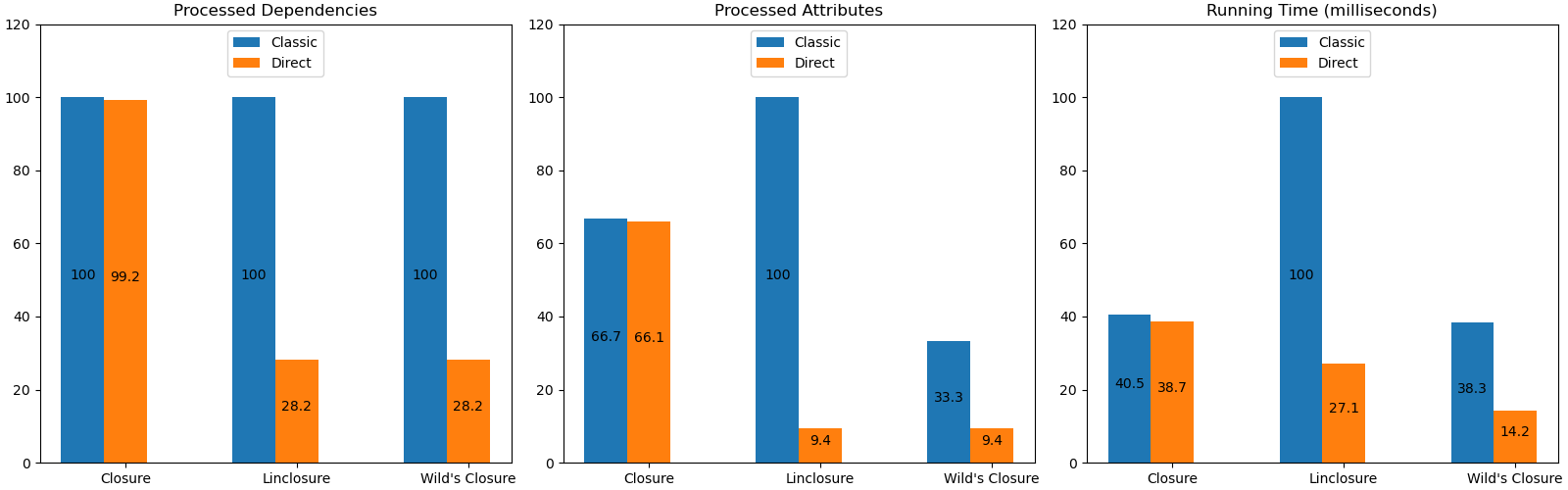}
	\caption{%
    Comparison of the performance of each algorithm w.r.t. their direct versions when processing the \cdub in \textbf{\sintetic datasets}.
    The values have been normalized to the interval (0,100).
    }
	\label{fig:comparacio_synthetic_CanonicalDirectUnitBasis}
\end{figure}
} 

\def\compDBASISsynthetic{ 
\begin{figure} 
	\centering
	\includegraphics[width=1.0\textwidth]{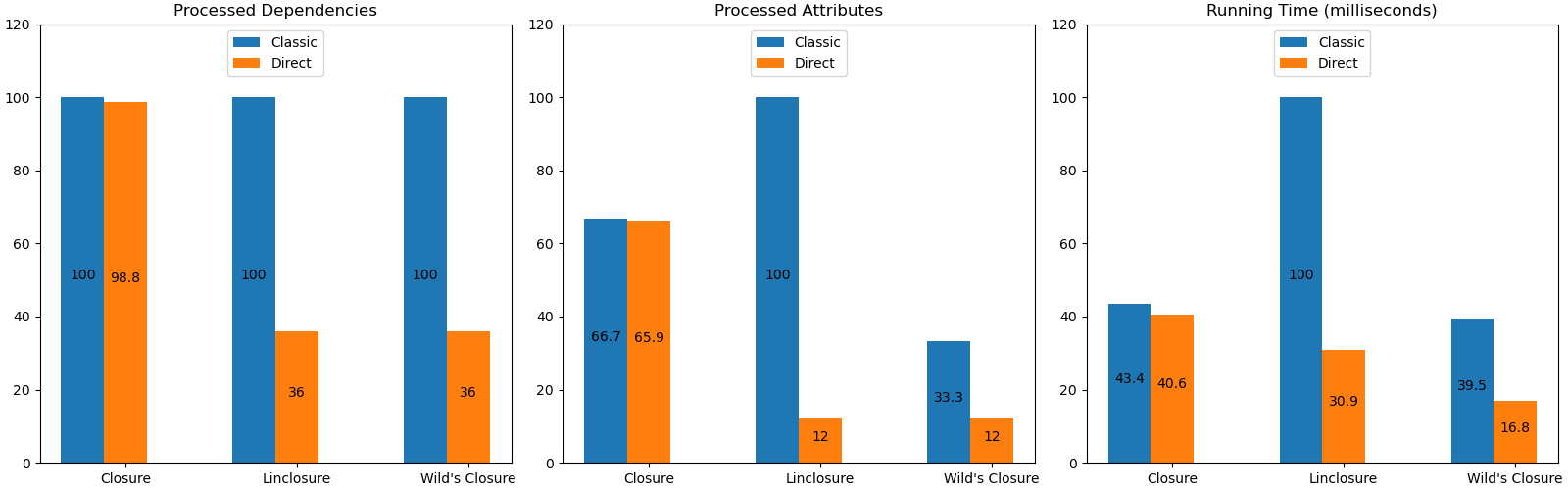}
	\caption{%
    Comparison of the performance of each algorithm w.r.t. their direct versions when processing the \dbasis in \textbf{\sintetic datasets}.
    The values have been normalized to the interval (0,100).
    }
	\label{fig:comparacio_synthetic_DBasis}
\end{figure}
} 

\def\barressynthetic{ 
\begin{figure} 
	\centering
	\includegraphics[width=1.0\textwidth]{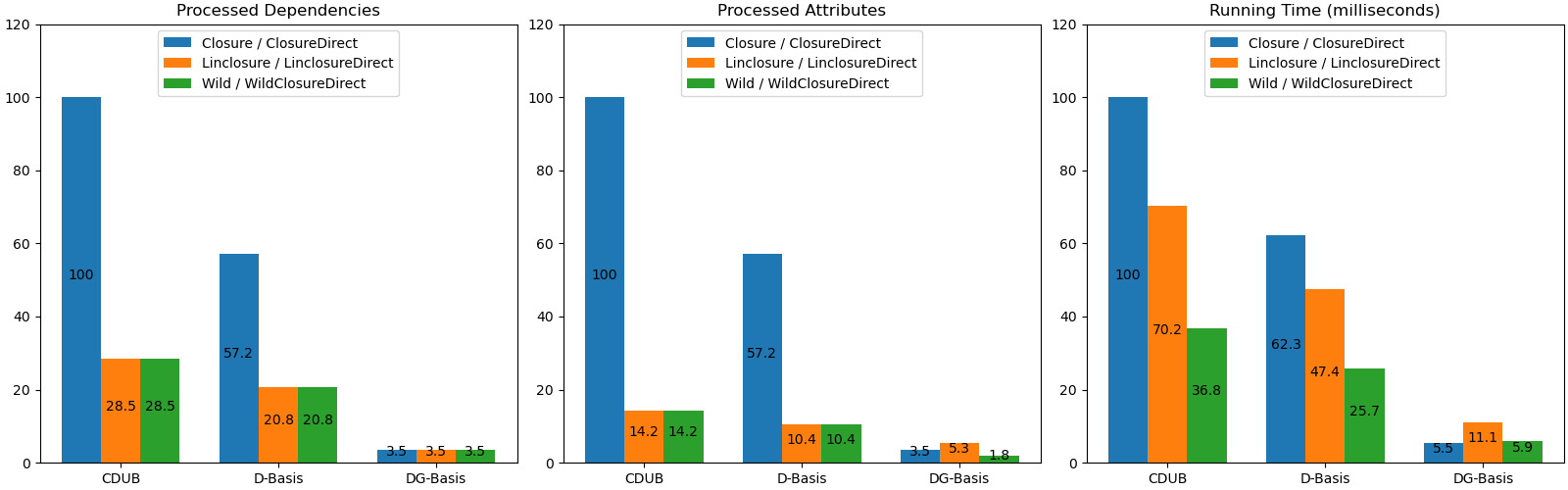}
	\caption{%
    Totals for the analyzed measures for each combination (Base $\times$ Algorithm) in \textbf{\sintetic datasets}.
    The \cdub and the \dbasis are processed with the \textit{direct} versions of the algorithms.
    The values have been normalized to the interval (0,100).
    }
	\label{fig:barres_syntheticdeps-atrib-in-out}
\end{figure}
} 

\def\barresNoDirectsynthetic{ 
\begin{figure} 
	\centering
	\includegraphics[width=1.0\textwidth]{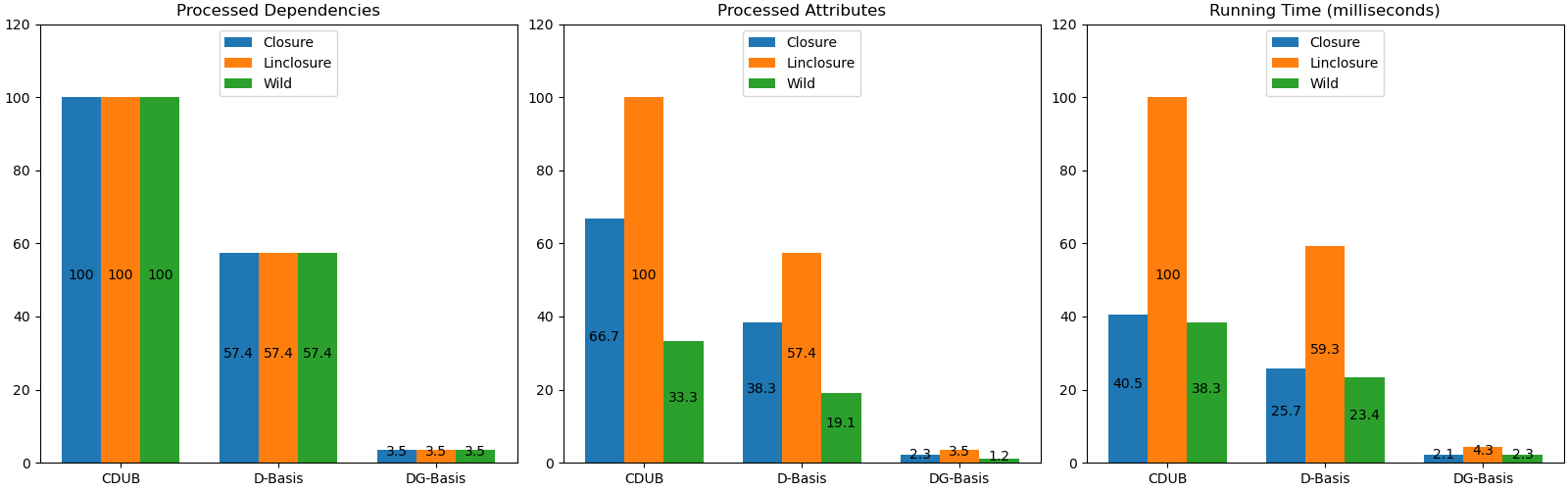}
	\caption{%
    Totals for the analyzed measures for each combination (Base $\times$ Algorithm) in \textbf{\sintetic datasets}.
    All bases are processed with the \textit{classical} versions of the algorithms.
    The values have been normalized to the interval (0,100).
    }
	\label{fig:barres_syntheticdeps-atrib-no-direct}
\end{figure}
} 

\def\liniessyntheticatribs{ 
\begin{figure} 
	\centering
	\includegraphics[width=1.0\textwidth]{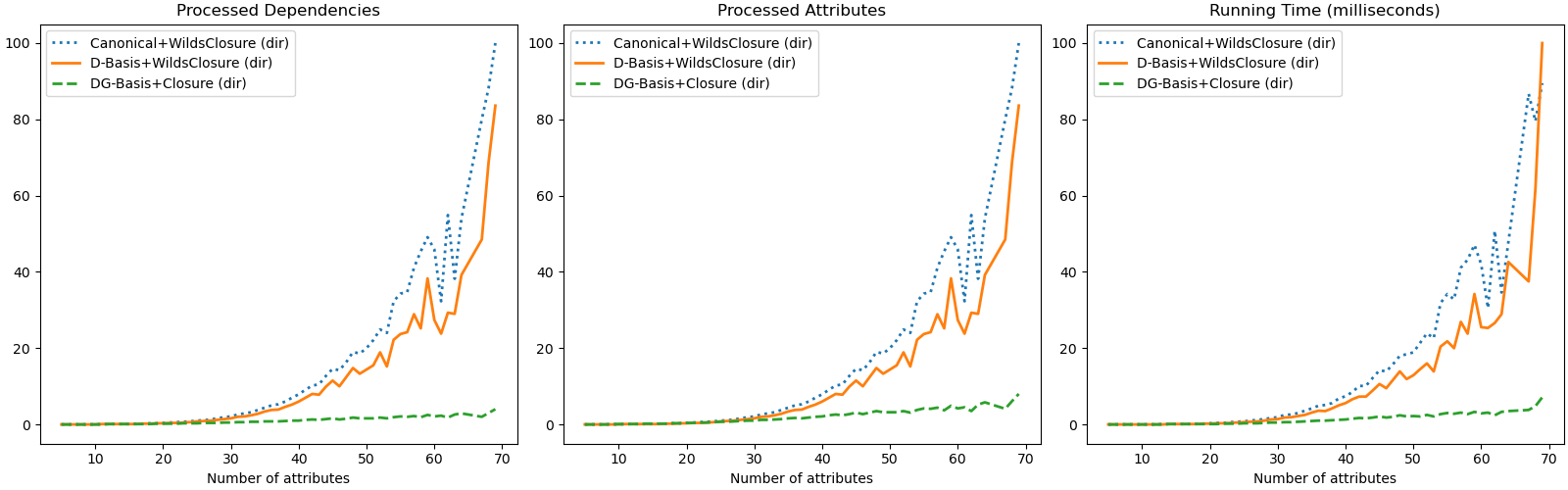}
	\caption{%
    Performance of the best combinations of (Base $\times$ Algorithm) for the analyzed metrics w.r.t. the number of attributes in \textbf{\sintetic datasets}.
    The values have been normalized to the interval (0,100).
    }
	\label{fig:linies_synthetic_liniessyntheticatribs}
\end{figure}
} 

\def\taulasynthetic{ 
\begin{table} 
\centering 
\resizebox{0.9\textwidth}{!}{ 
\begin{tabular}{|l||c|c|c||c|c|c||c|c|c|} 
\hline 
 & \multicolumn{3}{|c|}{\cdub} & \multicolumn{3}{|c|}{\dbasis} & \multicolumn{3}{|c|}{DG-basis} \\\hline 
 Attribute &  CLO & LIN & WILD & CLO & LIN & WILD & CLO & LIN & WILD \\ 
\hline 
\hline 
Processed Dependencies & 0 & 317 & 317 & 232 & 121 & 121 & 4879 & 4879 & 4879  \\ 
\hline 
Processed Attributes & 0 & 361 & 361 & 0 & 187 & 187 & 0 & 0 & 5001  \\ 
\hline 
Running Time (milliseconds) & 0 & 0 & 41 & 187 & 0 & 1044 & 4011 & 0 & 266  \\ 
\hline 
\hline 
\end{tabular} 
} 
\caption{%
Best performance in the \textbf{\sintetic datasets} for each pair base+algorithm for the 5 metrics.
The total number of databases is 5549 (for each metric there can be more than one minimal combination).
} 
\label{table:ranking_synthetic} 
\end{table} 
} 

\def\comparacioCDUBDGBasissynthetic{ 
\begin{figure} 
	\centering
	\includegraphics[width=0.7\textwidth]{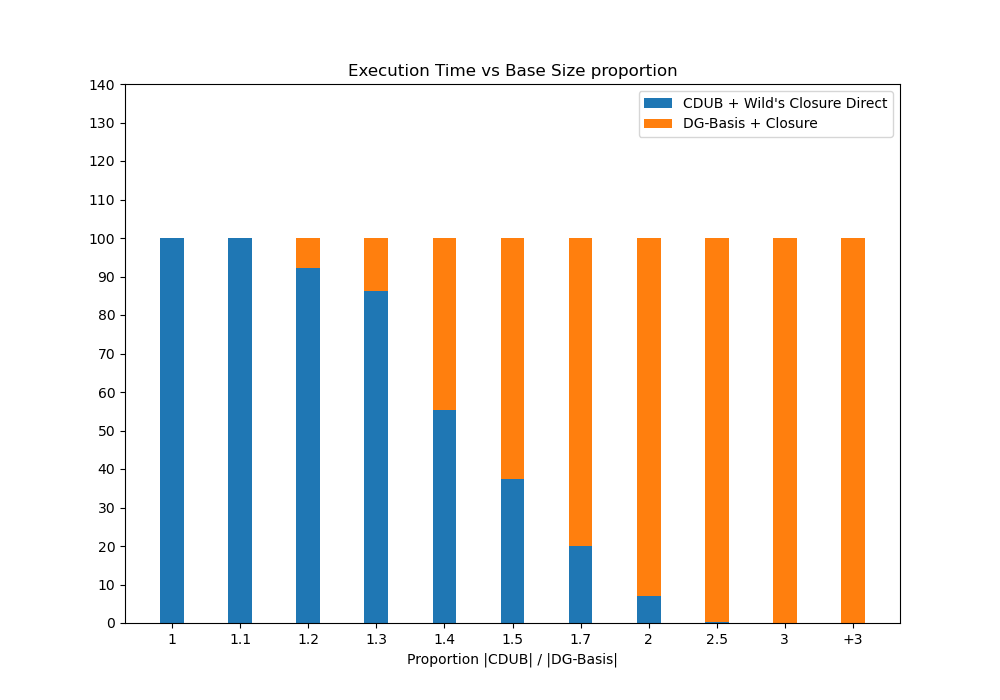}
	\caption{%
		X-axis: proportion $\sfrac{|\cdub|}{|DGBasis|}$. 
		Y-axis: proportion of combinations (Base $\times$ Algorithm) 
		that have the best execution time in \textbf{\sintetic datasets}. 
		All values have been normalized to the interval (0,100).
	}
	\label{fig:comparacioCDUBDGBasissynthetic}
\end{figure}
} 

\def\comparacioCDUBDBasesynthetic{ 
\begin{figure} 
	\centering
	\includegraphics[width=0.7\textwidth]{comparacio_CDUB_DBase_synthetic.png}
	\caption{
		X-axis: proportion $\sfrac{|\cdub|}{|\dbsis|}$. 
		Y-axis: proportion of combinations (Base $\times$ Algorithm) 
		that have the best execution time in \textbf{\sintetic datasets}. 
		All values have been normalized to the interval (0,100).
	}
	\label{fig:comparacioCDUBDBasesynthetic}
\end{figure}
} 



%
\section{Experiments}
\label{sec:experiments}
\subsection{Goals of the Present Experiments}
\label{sec:experiments_goals}

We take as a departure point and generalize the experiments performed in
both~\cite{BazhanovO14} and~\cite{AdarichevaNR13} as follows:

\begin{enumerate}

\item
The main question that we address is whether a \dgbasis can be competitive
--in terms of efficiency-- 
to compute $\clo$ compared to a direct basis, \cdub or \dbasis.

\item

We ensure that the algorithms process a direct basis to its full capacity.
Accordingly, we compare the performance of the three algorithms
and their \textit{direct} versions.

\item
We compare the performance of all possible combinations of the three algorithms
with the three bases
to ``control w.r.t. the algorithm'',
i.e., to prevent the fact that one algorithm may bias the results 
because it is more performing with a specific basis.

\item
We analyze different metrics, such as 
execution time, number of processed dependencies,
and number of attribute operations.

\item
We perform experiments over real and synthetic datasets. 

\end{enumerate}

\subsection{Experiments: Previous Work}
\label{sec:previousexperiments}

The first set of experiments in \cite{BazhanovO14} compares the performance 
of \closure, \linclosure, and \wild in the computation of $\clo$ with the \dgbasis.
The results show that 
``\textit{Algorithm 1 [Closure] was the fastest and Algorithm 2 [LinClosure] was the slowest}'',
which could be explained by the cost of the initialization step in \linclosure.
\wild ranks between \closure and \linclosure when applied to both synthetic and real datasets.

In another set of experiments, the 
authors fix a given number of dependencies ($1000$) and compute $\clo(X)$ with random $X$,
where the size of the set of attributes varies from $5,000$ to $100,000$.
In this case again, the execution time of \closure remains practically constant w.r.t. an increasing number of attributes, whereas the time grows linearly in both \linclosure and \wild.
The authors argue that 
``\textit{[T]he reason is that Algorithm 1 [\closure]
is quadratic in the number of implications, which is constant in this experiment}''.

In~\cite{AdarichevaNR13}, authors perform two types of experiments.
The first one consists in testing the performance of
\closure,
Forward Chaining Algorithm
--an algorithm used in Logics to check the satisfiability of Horn formulas \cite{DowlingG84}--,
and \wild.
They generate different D-bases including $5$ to $8$ attributes, 
and compare the \textit{execution time} of each algorithm.
It appears that \closure outperformed \wild in all these tests
with a small number of attributes.
However, the results also show that the difference in performance between both algorithms decreases when the number of attributes increases.
An important remark is that authors ensure that \closure performs only one single pass over $\Sigma$.
In another experiment authors generate different random closure systems, compute the \dgbasis and the \dbasis, and then compare the performance of both bases when computing $\clo$ using \closure.
The results show that the \dbasis checks fewer dependencies than the \dgbasis
on the average in experiments where the number of attributes is
$6$ and $7$.

\subsection{Datasets}
\label{sec:materials}

\begin{table}[ht]
\centering
\resizebox{1\textwidth}{!}{
\begin{tabular}{|l||c|c|c|c|c||l||c|c|c|c|c|}
\hline
Dataset & $|G|$ & $|M|$ & $|\Sigma_{\cdub}|$ & $|\Sigma_{\dbasis}|$ & $|\Sigma_{DG}|$ & 
Dataset & $|G|$ & $|M|$ & $|\Sigma_{\cdub}|$ & $|\Sigma_{\dbasis}|$ & $|\Sigma_{DG}|$ \\
\hline \hline
abalone & 64 & 9  & 54 &  54 & 40 & glass & 36 & 10 & 43 & 43 & 32 \\ 
\hline 
page-blocks & 27 & 11 & 80 & 80 & 34 & flights 20 500k & 23 & 12 & 42 & 30  & 28 \\ 
\hline 
atom sites & 41 & 12  & 79 & 79 & 55 & bridges & 63 & 12  & 87 & 87 & 62 \\ 
\hline 
wine & 76 & 14  & 159 & 159 & 131 & tax & 104 & 15  & 126 & 108 & 82 \\ 
\hline 
pen-recognition & 4088 & 17  & 9313 & 9313 & 6418 & ncvoter-1m-19 & 295 & 18  & 297 &  297 &  246 \\
\hline 
zoo & 80 & 18  & 249 & 248 & 143 & flight 1k 30c & 288 & 19  & 923 & 615 & 386 \\ 
\hline 
hepatitis & 608 & 20 & 5372 &  5372 &  2191 & soybean-small & 190 & 21  & 3157 & 2542 &  411 \\ 
\hline 
mushroom & 151 & 22 & 3067 & 3067 & 1508 & fd-reduced-1k-30 & 26 & 26  & 322 & 235 & 235  \\ 
\hline
fd-reduced-250k-30 & 313 & 26  & 2439 & 1545 & 1545 & automobile & 305 & 26 & 2676 & 2582 &  1350 \\ 
\hline 
anneal & 307 & 32  & 1265 & 1262 & 791 &  &  &  &  &  &  \\
\hline
\end{tabular}
}
\caption{%
The datasets \real from UCI Repository.
$|G|$ is the number of objects and $|M|$ is the number of attributes of
the related reduced and clarified contexts.
$|\Sigma_{\cdub}|$ is the size of the \cdub,
$|\Sigma_{\dbasis}|$ the size of the \dbasis,
and finally $|\Sigma_{DG}|$ the size of the \dgbasis.}
\label{table:realdatasets}
\end{table}

\textbf{Real datasets (\real)}.
We analyzed $19$ datasets from the 
\textit{UC Irvine Machine Learning Repository}\footnote{%
https://archive.ics.uci.edu/}.
Every dataset is transformed as a reduced and clarified context, possibly leading to a reduction of the initial number of attributes, and displayed in Table~\ref{table:realdatasets}.
For a given dataset, at least two bases have different sizes, meaning that
either
(a) the size of the \cdub is greater than the size of the \dgbasis,
or (b) the size of the \cdub is greater that the size of the \dbasis,
or (c) the size of the \dbasis is greater than the size of the \dgbasis,
or a combination of the above.

\textbf{Synthetic Datasets (\sintetic)}.
Reduced and clarified synthetic contexts were built in combining all possible values 
of the following parameters:
\begin{itemize}
   \item range of number of objects: 4 - 33.
    \item range of the number of attributes: 5 - 69.
    \item average number of files for each combination: 3.
    \item total number of files: 5550.
\end{itemize}

\subsection{Methodology}
\label{sec:methods}

\begin{table}[ht]
\centering
\resizebox{1.0\textwidth}{!}{
\begin{tabular}{|l||l|l|l|}
\hline
\cdub             & \closuredir & \linclosuredir & \wilddir \\ 
\hline
\dbasis           & \closuredir & \linclosuredir & \wilddir \\ 
\hline
\dgbasis          & \closure    & \linclosure    & \wild    \\  
\hline
\end{tabular}
}
\caption{Algorithms used to compute $\clo$ for each basis.}
\label{table:basis_algo_combinations}
\end{table}

We implemented the FDep algorithm~\cite{flachsavnik99} to compute the \cdub
and we implemented the \textit{Canonical Basis} algorithm~\cite{GanterO16}
to compute the \dgbasis,
using C++ and the \verb|dynamic_bitset| class of the \verb|boost| libraries.
The \dbasis is computed with the \verb|npar/dbasis|
algorithm\footnote{\url{https://gitlab.com/npar/dbasis}}. 
The algorithms computing $\clo$ for every basis are given in 
Table~\ref{table:basis_algo_combinations}.
We added the following counters to all algorithms:

\begin{enumerate}

\item
    \textbf{deps} counts the number of times a dependency
    is \textit{processed},
    i.e., used to compute $\clo(X)$,
    indicated by \texttt{deps} in every algorithm.

\item
    \textbf{attributes} counts the number of attribute operations,
    i.e., the number of times union, intersection, and
    difference are performed over sets of attributes.

\item
    \textbf{time} counts the number of milliseconds spent for computing
    $\clo(X)$.
    It should be noticed that \textbf{time} only counts
    the milliseconds strictly used for computing $\clo$ each time
    this function is called.
    
\end{enumerate}

Two additional counters are also examined
(see Appendices~\ref{app:realdatasets} and~\ref{app:syntheticdatasets}).

\begin{enumerate}

\item 
    \textbf{inner loop} counts the number of times
    an algorithm iterates in the inner loop, indicated by
    \texttt{inner loop} in the algorithm.

\item 
    \textbf{outer loop} counts the number of times
    an algorithm iterates in the outer loop (if available),
    indicated by \texttt{outer loop} in the algorithm.
   
\end{enumerate}

For each dataset in \real and \sintetic, we compute the closure
of $50,000$ random sets of attributes for every combination
\textit{\combinacio} in Table~\ref{table:basis_algo_combinations},
i.e., 9 combinations.
To ensure the precision of the execution time, 
each computation is repeated $100$ times and then the average is considered.

In \linclosure and \wild, and in their direct versions,
we do not count the part remaining constant,
i.e., the initialization of the variable $list$,
and we assume that $\Sigma$ does not change during a processing.
In \linclosure and \linclosuredir, we count the update of the variable
$count$ that needs to be recomputed each time $\clo$ is called.
The computation of $\clozero$, being necessary when $\Sigma$ is a \dbasis
in \linclosuredir and \wilddir, is precomputed and therefore not counted.

All tests were run with the cluster facilities of the High Performance Computing Center at UPC\footnote{\url{https://rdlab.cs.upc.edu/hpc/}},
ensuring that each execution is performed in an isolated environment 
with a dedicated Intel(R) Xeon(R) CPU X5650 at 2667 MHz and 1 Gb of memory.
This guarantees that all combinations are computed in the same conditions.

\subsection{Format of the Results}
\label{sec:format}

For every \real and \sintetic dataset we have summarized the results
of every measure,
i.e., processed dependencies, processed attributes, and running time.
For every combination \combinacio these results are plotted in 
Figures~\ref{fig:barres_realdeps-atrib-in-out}
and~\ref{fig:barres_syntheticdeps-atrib-in-out}.
Let us explain the content of these plots when calculating $\clo$
with a combination \combinacio and processing a dataset in \real.
For every dataset in $\real = \conjunt{D_1, D_2, \dots, D_{19}}$,
we compute $\clo$ for $50,000$ random sets of attributes
and we sum up all the processed dependencies,
i.e., $\deps(D_i) = \sum \deps(\clo(X))$,
where $\deps(\clo(X))$ denotes the number of processed dependencies.

Finally, we average all $\sum_{D_i \in \textbf{real}} \deps(D_i)$
for all combinations, 
leading to a \textit{grand total} for each of the nine combinations \combinacio. 
We repeated each computation $100$ times and averaged the values of all considered variables.
We normalized  to the interval $(0,100)$ and then plotted these grand totals. 

Regarding the \sintetic datasets, we also computed the evolution of the metrics
w.r.t. the number of attributes.
We grouped all the datasets with the same number of attributes  and computed the average for each metric.
We plotted the results in Figure~\ref{fig:linies_synthetic_liniessyntheticatribs}.
Here we only compare the best performing combinations \combinacio for every basis.

We also computed a ranking table recording how many times a combination \combinacio 
is the best performer in the computation of each metric.
The results are presented in Tables~\ref{table:ranking_real} and~\ref{table:ranking_synthetic}.
In particular, in Table~\ref{table:ranking_real},
every column is a combination \combinacio,
and every row is one of the computed metrics. 
For example, the score of \textit{Processed Dependencies} (first row)
for the combination \combi{\cdub}{\linclosuredir} is $10$.
This means that the combination \combi{\cdub}{\linclosuredir} is the best performer 
when computing \textbf{Processed Attributes} in $10$ \real datasets.
Since the total number of datasets in \real is $19$,
each row must sum, at least, $19$, but it may be larger,  as there can be ties,
i.e., there can be more than one \textit{winning combination}.

All the numerical results are presented in the technical report
\cite{DBLP:journals/corr/abs-2404-12229} (Appendices A and B).



%
\section{Results}
\label{sec:results}
\subsection{Impact of the Direct Versions of the Algorithms}
\label{subsec:validation}
%

We present an experimental comparison of the \textit{direct} versions of the 
algorithms, guided by the following question:
\textit{``How the changes performed in \closure, \linclosure, and \wild are relevant''?}
Recall that the theoretical differences are discussed in
Section~\ref{sec:impactcomplexity}.

Figures~\ref{fig:comparacio_real_CanonicalDirectUnitBasis}
and~\ref{fig:comparacio_synthetic_CanonicalDirectUnitBasis}
show the differences between the values of the indicators \texttt{deps}, \texttt{attributes}, and \texttt{time}, using the classical and direct versions of every algorithm, when processing the \cdub over \real and \sintetic datasets respectively.
Figures~\ref{fig:comparacio_real_DBasis} and~\ref{fig:comparacio_synthetic_DBasis}
show the same results when processing the \dbasis.

\compCANONICALreal
\compCANONICALsynthetic

Focusing on the running time and \real datasets, we can observe 
that algorithms \linclosure and \wild almost double their processing time w.r.t. 
their \textit{direct} versions, 
whereas \closuredir shows a slight improvement over \closure when processing both 
the \cdub and the \dbasis.

In \sintetic datasets the same tendency is observed, 
while the improvement of \linclosuredir and \wilddir over their non-direct counterparts
is even more acute.
Such results clearly show that running the direct version of an algorithm with a direct basis is much more efficient that running its classical version.

\compDBASISreal
\compDBASISsynthetic

\subsection{The \dgbasis is an Efficient Alternative}

\barresreal
\barressynthetic

Based on the results of the experiments, 
we now propose an answer to the main question formulated in the Introduction (\S~\ref{sec:introduction}):
\textit{``Is the \dgbasis more efficient than the \cdub and \dbasis (direct bases)
to compute the closure of a set of attributes?''}

Firstly, let us check the grand totals for all pairs \combinacio involving
\real and \sintetic datasets respectively and shown in
Figures~\ref{fig:barres_realdeps-atrib-in-out}
and~\ref{fig:barres_syntheticdeps-atrib-in-out}.
These totals provide a slightly different answer to the question.
Regarding the running time,
there are two most efficient combinations,
namely \combi{\cdub}{\wilddir} and \combi{\dbasis}{\wilddir} for \real datasets,
and \combi{\dgbasis}{\closure} and \combi{\dgbasis}{\wild} for \sintetic datasets.
However, in \real datasets, the difference between the direct bases and
the best cost with direct bases is $33$ versus $44$ for the \dgbasis
(w.r.t. a scale $0 - 100$).
In \sintetic datasets, the best cost is $6$ for the \dgbasis versus $25$ for the best combination involving a direct basis, i.e., \mbox{\combi{\closure}{\wilddir}}.
To sum up, in \real datasets the best performance is in favor of direct bases,
while in \sintetic datasets the best performance is in favor of the \dgbasis
in a \textit{much larger} proportion.

\taulareal

Now we examine Tables~\ref{table:ranking_real} and~\ref{table:ranking_synthetic},
where the \textbf{counting} of the best performing combinations is shown.
In \real datasets, the best performing combinations are \combi{\cdub}{\wilddir}
and \combi{\dbasis}{\wilddir} (see Table~\ref{table:ranking_real}).
There are also $6$ files, namely
\textit{page-blocks, atom-sites, soybean, hepatitis, mushrooms and anneal}
(\cite{DBLP:journals/corr/abs-2404-12229} Appendix A),
where the best combination is \combi{\dgbasis}{\closure}.
In these files the size of the \dgbasis is significantly smaller than the corresponding \cdub and \dbasis.

\taulasynthetic

In \sintetic datasets, the \dgbasis is the best performing basis in the vast 
majority of cases, actually in $4277 \, (4011+266)$ out of $5549$ datasets
while the \dbasis performs better in $1231 \, (1044+187)$ datasets,
and the \cdub only in $41$ datasets
(see Table~\ref{table:ranking_synthetic}).
A large amount of files was analyzed in the \sintetic datasets and
Figure~\ref{fig:linies_synthetic_liniessyntheticatribs} shows the tendency of
the three bases combined with the best performing algorithm.
It can be observed that \combi{\dgbasis}{\closure} reveals a steady tendency contrasting the two other combinations involving the \cdub and the \dbasis.

To summarize, we can observe that the \dgbasis is the best performing alternative 
--in terms of execution time-- in $82\%$ of \sintetic datasets,
i.e., a large majority,
whereas it is the best alternative in only $31\%$ of \real datasets,
which is not so weak.

\liniessyntheticatribs

\subsection{The Impact of the Basis Size}

The preceding results show that the \dgbasis is actually an alternative
to both \cdub and \dbasis.
We now try to characterize when this is the case, i.e.,
\textit{in which cases is the \dgbasis more efficient than both the \cdub and the \dbasis}.
As it can be expected, a main factor when comparing a minimal basis and direct bases 
is their ``relative size'', i.e., the proportion of the sizes of, say, 
the \cdub and the \dgbasis for a given dataset.
Below we show that this metric is relevant to answer the question above.

\comparacioCDUBDGBasissynthetic

Figure~\ref{fig:comparacioCDUBDGBasissynthetic} shows the proportion of
combinations \combi{\dgbasis}{\closure} and \mbox{\combi{\cdub}{\wilddir}} 
that are the most efficient w.r.t. computational time for each basis.
The X-axis indicates the size proportion |\cdub|/|\dgbasis|. 
For example, the bar $X = 1.4$ counts all the \sintetic datasets
such that $1.3 < \frac{|\cdub|}{|\dgbasis|} \leq 1.4$.
The blue and orange colors show the proportion of cases in which one of the
combinations \combi{\dgbasis}{\closure} (orange) or \combi{\cdub}{\wilddir} (blue)
is the best performing (w.r.t. a scale of $0 - 100$).

Stated differently, the bar $X = 1.4$ tells us that whenever 
the \cdub size is between $1.3$ and $1.4$ times larger than the related \dgbasis size, the pair \combi{\cdub}{\wilddir} is the best performing combination w.r.t. 
running time in $53\%$ of the cases.

By contrast, when $X = 2$, the pair \combi{\dgbasis}{\closure} is the best 
performing combination in $95\%$ of the cases.
Thus, as soon as the size of the \dgbasis is the half of the size of the \cdub, 
the \dgbasis is more efficient for computing the closure of a set of attributes 
when \sintetic datasets are processed.
These results can be observed as well in \real datasets.

\subsection{Impact of the Selected Algorithm}
We have discussed the behavior of the three bases in combination
with one of the three different algorithms to compute the closure of a set of attributes.
The next question is to verify the influence of the closure algorithm used in a combination:
\textit{``Does the performance of a basis depend on the algorithm which is used?''}
In Figures~\ref{fig:barres_realdeps-atrib-in-out}
and~\ref{fig:barres_syntheticdeps-atrib-in-out},
the performance of the three \textbf{direct} versions of the algorithms
\closuredir, \linclosuredir, and \wilddir,
in combination with the direct bases \cdub and \dbasis 
roughly corresponds to the theory: 
\wild (\wilddir) is an improvement
of \linclosure (\linclosuredir) which, in turn,
is an improvement of \closure (\closuredir).

By contrast, \wild or \closure algorithms show a similar performance
when combined with the \dgbasis,
whereas \linclosure is slightly less efficient.
More precisely, we can observe in Table~\ref{table:ranking_synthetic} that
the best performing algorithm is \closure with $4011$ best cases versus $266$
for \wild.
This difference justifies the use of all three algorithms in our experiments,
since this allows to compare the best options \combinacio in each case
and to prevent any bias w.r.t. an algorithm.

Since the algorithms used in combination with the \dgbasis are the 
\textit{original} versions without any modification, one possible explanation
of this difference could be that the changes performed in \linclosuredir and in \wilddir are more effective than those performed in \closuredir.
This can be seen in Figures~\ref{fig:barres_realdeps-atrib-in-out} 
and~\ref{fig:barres_syntheticdeps-atrib-in-out},
where the number of processed dependencies by \closuredir is between two and three times --in \real and \sintetic datasets respectively-- larger than the number processed by \linclosuredir and \wilddir.
It can be inferred that the number of processed dependencies is indeed a determining factor of the asymptotic complexity of these algorithms.

\barresNoDirectreal
\barresNoDirectsynthetic

A related question is
\textit{``What happens when the three bases are processed with the classical version of the algorithms?''}
Figures~\ref{fig:barres_realdeps-atrib-no-direct} and~\ref{fig:barres_syntheticdeps-atrib-no-direct}
show the comparison between the three bases when DBLP:journals/corr/abs-2404-12229they are processed with the \textit{classical} 
versions of \closure, \linclosure, and \wild.
We can observe that, overall, the best performing basis is the \dgbasis in both \real and \sintetic datasets.

Again these results justify the use of the adapted \textit{direct} versions of the algorithms since they make the direct bases much more competitive w.r.t the \dgbasis.

\subsection{Impact of the Number of Processed Dependencies and Attribute Operations}

To terminate, let us examine a final question
\textit{``Besides the running time, do the number of processed dependencies and the number of attribute operations have any influence when processing the bases?''}.
These two important metrics are significant because:
(i) the number of processed dependencies is of capital importance in analyzing the asymptotic complexity of the algorithms,
and (ii) the number of attribute operations, although usually neglected in such a complexity analysis, is mentioned in~\cite{Wild95} as a relevant factor in the running time of the algorithms:
\textit{``Doing few set operations with big sets is better than doing
many set operations with small sets.''}

Table~\ref{table:correlations} shows the correlations between
\textit{attribute operations} and \textit{processed dependencies} 
w.r.t. \textit{running time}.
These high values are somehow expected from the theoretical complexity analysis of the algorithms.
If we examine the grand totals in Figures~\ref{fig:barres_realdeps-atrib-in-out} 
and~\ref{fig:barres_syntheticdeps-atrib-in-out},
we can also observe similar tendencies for all three variables.

\begin{table} 
\centering 
\resizebox{\textwidth}{!}{ 
\begin{tabular}{|l||c|c|c||c|c|c|} 
\hline
 &  \multicolumn{3}{c|}{Real} & \multicolumn{3}{c|}{Synthetic} \\
\hline
 & Closure & LinClosure & WildClosure  & Closure & LinClosure & WildClosure \\
\hline \hline
dependencies & 0.99 & 0.97 & 0.99 & 0.99 & 0.98 & 0.99 \\
\hline
attribute & 0.99 & 0.97 & 0.99 & 0.99 & 0.97 & 0.99 \\
\hline
\end{tabular} 
} 
\caption{%
Spearman correlations of the metrics 
\textbf{number of dependencies, attribute operations} w.r.t. the metric
\textbf{execution time} for each used algorithm.
All the correlations are significant ($p-value << 0.05$).
} 
\label{table:correlations} 
\end{table} 

Finally, all these results are witnessing the fact that the \dgbasis can efficiently 
compete with the two direct bases \cdub and \dbasis, even when the latter are processed 
at their full capacity, i.e., using the direct versions of the algorithms computing a closure.



\section{Discussion and Conclusion}
\label{sec:discussion}

We performed a series of experiments over different datasets to answer one main question,
i.e., \textit{``Is it better to use a direct basis or a minimal basis
to compute a closure $\clo$''}?
The results show a contrasted answer as the \cdub is the best option in \real
while the \dgbasis is the best option in \sintetic datasets.
The size proportion of the bases may explain such results.
However, a crucial fact is that the \cdub and the \dbasis should be processed with a \textit{direct} algorithm, otherwise the best option in all cases remains the \dgbasis.
This makes the \dgbasis a quite strong and valuable alternative to the \cdub and the \dbasis in many situations.

In the experiments we also combine the three bases with the three closure algorithms to prevent any algorithm bias, i.e., the use of a single algorithm could favor one of the bases.
Indeed, it is observed that the best performing algorithm working with the \dgbasis
is \closure, whereas it is \wilddir with the two direct bases.

Before concluding, let us examine two last remarks.
Firstly, all results are valid provided that the bases do not change over the time.
How would the nine combinations \combinacio change when the bases dynamically change is an open question.
Secondly, in this paper ``bases of dependencies'' are processed rather than
``random sets of dependencies''.
This important difference involves some restrictions in the computing of the bases,
and in particular the size of the attribute set cannot be arbitrarily too large.

To sum up, the answer to the main question
\textit{``Is the \dgbasis a competitive alternative to compute the closure $\clo$''}
is positive provided that the size of the \dgbasis is sufficiently smaller than the 
sizes of the \cdub or the \dbasis.
As it can be observed that the sizes of \cdub, \dbasis, and \dgbasis,
tend to be similar in many real datasets, the \dgbasis appears to be a valuable 
alternative w.r.t. the \cdub and the \dbasis

Regarding future work, this study should be continued in many directions as shown by the questions raised in the experiments, among which:
given a dataset, can we design relevant metrics able to decide when a \dgbasis --or a \dbasis-- 
will contain significantly sufficiently less dependencies than the equivalent \cdub, 
leading to a good combination choice for computing a closure.
Similarly, more insight should be provided about the impact of the number of attributes and the size of the datasets on the performance of the different combinations that are tested.
Finally, a next step is to publicize the capabilities of the \dgbasis in the DB community 
based on the results and the analysis presented in this present paper.
This is indeed one main future work.

\section{Acknowledgements}

Jaume Baixeries is supported by a recognition 2021SGR-Cat (01266 LQMC) 
from AGAUR (Generalitat de Catalunya) and the
grants AGRUPS-2022 and AGRUPS-2023 from Universitat Politècnica de Catalunya.
Amedeo Napoli is carrying out this research work as part of the French ANR-21-CE23-0023 SmartFCA Research Project.


\bibliographystyle{plain}
\bibliography{bibliography}


\appendix

\section{Experiments with Real Datasets}
\label{app:realdatasets}


\begin{table} 
\centering 
\resizebox{\textwidth}{!}{ 

\begin{tabular}{|l|c||c|c|c|c||c|c|c|c||c|c|c|c|} 
\hline 
\multicolumn{2}{|c|}{\textbf{Processed Dependencies}} & \multicolumn{4}{|c|}{CDUB} & \multicolumn{4}{|c|}{D-Base} & \multicolumn{4}{|c|}{DG-basis} \\\hline 
 DB & $|\mathcal{U}|$ & $|\Sigma|$ & CLO & LIN & WILD & $|\Sigma|$ & CLO & LIN & WILD & $|\Sigma|$ & CLO & LIN & WILD \\ 
\hline 
\hline 
abalone & 9  & 54 & 1,113,556.04  & \textbf{525,225.74}  & \textbf{525,225.74}  & 54 & 1,113,556.04  & 750,176.85  & 750,176.85  & 40 & 947,103.83  & 947,103.83  & 947,103.83 \\ 
\hline 
glass & 10  & 43 & 1,151,997.54  & \textbf{566,797.07}  & \textbf{566,797.07}  & 43 & 1,151,997.54  & 816,795.43  & 816,795.43  & 32 & 939,873.76  & 939,873.76  & 939,873.76 \\ 
\hline 
page-blocks & 10  & 80 & 1,867,705.26  & 955,963.98  & 955,963.98  & 80 & 1,867,705.26  & 1,205,925.99  & 1,205,925.99  & 34 & \textbf{899,487.21}  & \textbf{899,487.21}  & \textbf{899,487.21} \\ 
\hline 
flights 20 500k & 12  & 42 & 783,985.25  & \textbf{565,266.00}  & \textbf{565,266.00}  & 30 & 577,850.04  & 738,738.73  & 738,738.73  & 28 & 568,534.11  & 568,534.11  & 568,534.11 \\ 
\hline 
atom sites & 12  & 79 & 1,174,221.89  & 976,250.03  & 976,250.03  & 79 & 1,174,221.89  & 1,276,199.38  & 1,276,199.38  & 55 & \textbf{886,779.01}  & \textbf{886,779.01}  & \textbf{886,779.01} \\ 
\hline 
bridges & 12  & 87 & 1,936,568.51  & \textbf{873,481.47}  & \textbf{873,481.47}  & 87 & 1,936,568.51  & 1,186,976.12  & 1,186,976.12  & 62 & 1,411,123.26  & 1,411,123.26  & 1,411,123.26 \\
\hline 
\end{tabular} 
 
} 
\caption{Totals of the measure \textbf{Processed Dependencies} for all \textbf{\real datasets}. In bold are the minimal values. The last line contains the average of each measure: the sum of all values for each pair (Base $\times$ Algorithm) divided by the number of datasets.} 
\label{table:total_deps_real} 
\end{table}

\begin{table} 
\centering 
\resizebox{\textwidth}{!}{ 

\begin{tabular}{|l|c||c|c|c|c||c|c|c|c||c|c|c|c|} 
\hline 
\multicolumn{2}{|c|}{\textbf{Processed Attributes}} & \multicolumn{4}{|c|}{CDUB} & \multicolumn{4}{|c|}{D-Base} & \multicolumn{4}{|c|}{DG-basis} \\\hline 
 DB & $|\mathcal{U}|$ & $|\Sigma|$ & CLO & LIN & WILD & $|\Sigma|$ & CLO & LIN & WILD & $|\Sigma|$ & CLO & LIN & WILD \\ 
\hline 
\hline 
abalone & 9  & 54 & 2,227,112.08  & \textbf{525,225.74}  & \textbf{525,225.74}  & 54 & 2,227,112.08  & 750,176.85  & 750,176.85  & 40 & 1,894,207.66  & 2,841,311.49  & 947,103.83 \\ 
\hline 
glass & 10  & 43 & 2,303,995.08  & \textbf{566,797.07}  & \textbf{566,797.07}  & 43 & 2,303,995.08  & 816,795.43  & 816,795.43  & 32 & 1,879,747.52  & 2,819,621.28  & 939,873.76 \\ 
\hline 
page-blocks & 10  & 80 & 3,735,410.52  & 955,963.98  & 955,963.98  & 80 & 3,735,410.52  & 1,205,925.99  & 1,205,925.99  & 34 & 1,798,974.42  & 2,698,461.63  & \textbf{899,487.21} \\ 
\hline 
flights 20 500k & 12  & 42 & 1,567,970.50  & \textbf{565,266.00}  & \textbf{565,266.00}  & 30 & 1,155,700.08  & 738,738.73  & 738,738.73  & 28 & 1,137,068.22  & 1,705,602.33  & 568,534.11 \\ 
\hline 
atom sites & 12  & 79 & 2,348,443.78  & 976,250.03  & 976,250.03  & 79 & 2,348,443.78  & 1,276,199.38  & 1,276,199.38  & 55 & 1,773,558.02  & 2,660,337.03  & \textbf{886,779.01} \\ 
\hline 
bridges & 12  & 87 & 3,873,137.02  & \textbf{873,481.47}  & \textbf{873,481.47}  & 87 & 3,873,137.02  & 1,186,976.12  & 1,186,976.12  & 62 & 2,822,246.52  & 4,233,369.78  & 1,411,123.26 \\
\hline 
\end{tabular} 
 
} 
\caption{Totals of the measure \textbf{Processed Attributes} for all \textbf{\real datasets}. In bold are the minimal values. The last line contains the average of each measure: the sum of all values for each pair (Base $\times$ Algorithm) divided by the number of datasets.} 
\label{table:total_atributs_real} 
\end{table}

\begin{table} 
\centering 
\resizebox{\textwidth}{!}{ 

\begin{tabular}{|l|c||c|c|c|c||c|c|c|c||c|c|c|c|} 
\hline 
\multicolumn{2}{|c|}{\textbf{Running Time (milliseconds)}} & \multicolumn{4}{|c|}{CDUB} & \multicolumn{4}{|c|}{D-Base} & \multicolumn{4}{|c|}{DG-basis} \\\hline 
 DB & $|\mathcal{U}|$ & $|\Sigma|$ & CLO & LIN & WILD & $|\Sigma|$ & CLO & LIN & WILD & $|\Sigma|$ & CLO & LIN & WILD \\ 
\hline 
\hline 
abalone & 9  & 54 & 87.41  & 95.88  & \textbf{64.01}  & 54 & 87.15  & 100.82  & 69.50  & 40 & 78.25  & 233.76  & 93.18 \\ 
\hline 
glass & 10  & 43 & 89.55  & 90.87  & \textbf{61.17}  & 43 & 89.35  & 95.28  & 66.80  & 32 & 79.50  & 234.96  & 97.13 \\ 
\hline 
page-blocks & 10  & 80 & 147.94  & 154.29  & 102.48  & 80 & 147.59  & 158.80  & 108.67  & 34 & \textbf{81.97}  & 229.30  & 101.23 \\ 
\hline 
flights 20 500k & 12  & 42 & 66.46  & 91.45  & 60.88  & 30 & \textbf{49.33}  & 77.90  & 53.69  & 28 & 57.57  & 155.25  & 76.67 \\ 
\hline 
atom sites & 12  & 79 & 108.75  & 155.47  & 103.56  & 79 & 109.10  & 161.46  & 109.79  & 55 & \textbf{98.88}  & 247.16  & 123.32 \\ 
\hline 
bridges & 12  & 87 & 152.03  & 157.79  & \textbf{103.37}  & 87 & 151.86  & 166.26  & 110.29  & 62 & 125.39  & 363.48  & 158.52 \\
\hline 
\end{tabular} 
 
} 
\caption{Totals of the measure \textbf{Running Time (milliseconds)} for all \textbf{\real datasets}. In bold are the minimal values. The last line contains the average of each measure: the sum of all values for each pair (Base $\times$ Algorithm) divided by the number of datasets.} 
\label{table:total_temps_real} 
\end{table}

\begin{table} 
\centering 
\resizebox{\textwidth}{!}{ 
\begin{tabular}{|c||l|l||c|c||c||c|c|c|c|c|c|c|c|} 
\hline 
 DB & Base & Algorithm & deps & attrib & outer & inner & time (ms) & $|\Sigma|$ & $|\mathcal{U}|$ \\ 
\hline 
\hline 
\multirow{11}{*}{abalone} 
 &  \multirow{3}{*}{Canonical} & Closure (dir) & 1,113,556 & 2,227,112 & \textbf{50,000} & 2,700,000 & 87.41 & 54 & 9 \\ \cline{3-10} 
 &  & Linclosure (dir) & \textbf{525,226} & \textbf{525,226} & 224,951 & 5,848,664 & 95.88 & 54 & 9 \\ \cline{3-10} 
 &  & WildsClosure (dir) & \textbf{525,226} & \textbf{525,226} & \textbf{50,000} & \textbf{525,226} & \textbf{64.01} & 54 & 9 \\ \cline{3-10} 
 \cmidrule[2pt]{2-10} 
 &  \multirow{3}{*}{D-Base} & Closure (dir) & 1,113,556 & 2,227,112 & \textbf{50,000} & 2,700,000 & 87.15 & 54 & 9 \\ \cline{3-10} 
 &  & Linclosure (dir) & 750,177 & 750,177 & 224,951 & 5,848,664 & 100.82 & 54 & 9 \\ \cline{3-10} 
 &  & WildsClosure (dir) & 750,177 & 750,177 & \textbf{50,000} & \textbf{525,226} & 69.50 & 54 & 9 \\ \cline{3-10} 
 \cmidrule[2pt]{2-10} 
 &  \multirow{3}{*}{DG-Basis} & Closure (dir) & 947,104 & 1,894,208 & 91,874 & 2,256,231 & 78.25 & 40 & 9 \\
\hline 
\end{tabular} 
} 
\caption{Total values of \textbf{\real datasets} per all analyzed measures: number of dependencies processed, number of operations on attributes, outer loops, inner loops  and computation time in miliseconds. $|\Sigma|$: size of the base. $|\mathcal{U}|$: number of attributes} 
\label{table:total_real_4} 
\end{table}

\begin{table} 
\centering 
\resizebox{\textwidth}{!}{ 
\begin{tabular}{|c||l|l||c|c||c||c|c|c|c|c|c|c|c|} 
\hline 
 DB & Base & Algorithm & deps & attrib & outer & inner & time (ms) & $|\Sigma|$ & $|\mathcal{U}|$ \\ 
\hline 
\hline 
\multirow{11}{*}{wine} 
 &  \multirow{3}{*}{Canonical} & Closure (dir) & 5,895,466 & 11,790,932 & \textbf{50,000} & 7,950,000 & 423.42 & 159 & 14 \\ \cline{3-10} 
 &  & Linclosure (dir) & \textbf{2,206,470} & \textbf{2,206,470} & 349,990 & 10,724,908 & 322.16 & 159 & 14 \\ \cline{3-10} 
 &  & WildsClosure (dir) & \textbf{2,206,470} & \textbf{2,206,470} & \textbf{50,000} & \textbf{2,206,470} & \textbf{215.55} & 159 & 14 \\ \cline{3-10} 
 \cmidrule[2pt]{2-10} 
 &  \multirow{3}{*}{D-Base} & Closure (dir) & 5,895,466 & 11,790,932 & \textbf{50,000} & 7,950,000 & 424.45 & 159 & 14 \\ \cline{3-10} 
 &  & Linclosure (dir) & 2,556,459 & 2,556,459 & 349,990 & 10,724,908 & 324.46 & 159 & 14 \\ \cline{3-10} 
 &  & WildsClosure (dir) & 2,556,459 & 2,556,459 & \textbf{50,000} & \textbf{2,206,470} & 227.10 & 159 & 14 \\ \cline{3-10} 
 \cmidrule[2pt]{2-10} 
 &  \multirow{3}{*}{DG-Basis} & Closure (dir) & 5,249,889 & 10,499,778 & 114,763 & 7,180,702 & 404.25 & 131 & 14 \\
\hline 
\end{tabular} 
} 
\caption{Total values of \textbf{\real datasets} per all analyzed measures: number of dependencies processed, number of operations on attributes, outer loops, inner loops  and computation time in miliseconds. $|\Sigma|$: size of the base. $|\mathcal{U}|$: number of attributes} 
\label{table:total_real_5} 
\end{table}

\begin{table} 
\centering 
\resizebox{\textwidth}{!}{ 
\begin{tabular}{|c||l|l||c|c||c||c|c|c|c|c|c|c|c|} 
\hline 
 DB & Base & Algorithm & deps & attrib & outer & inner & time (ms) & $|\Sigma|$ & $|\mathcal{U}|$ \\ 
\hline 
\hline 
\multirow{11}{*}{hepatitis} 
 &  \multirow{3}{*}{Canonical} & Closure (dir) & 133,332,430 & 266,664,860 & \textbf{50,000} & 268,600,000 & 16,800.33 & 5372 & 20 \\ \cline{3-10} 
 &  & Linclosure (dir) & \textbf{41,071,218} & \textbf{41,071,218} & 499,733 & 791,475,447 & 16,374.89 & 5372 & 20 \\ \cline{3-10} 
 &  & WildsClosure (dir) & \textbf{41,071,218} & \textbf{41,071,218} & \textbf{50,000} & \textbf{41,071,218} & 6,815.55 & 5372 & 20 \\ \cline{3-10} 
 \cmidrule[2pt]{2-10} 
 &  \multirow{3}{*}{D-Base} & Closure (dir) & 133,332,430 & 266,664,860 & \textbf{50,000} & 268,600,000 & 16,206.64 & 5372 & 20 \\ \cline{3-10} 
 &  & Linclosure (dir) & 41,570,951 & 41,570,951 & 499,733 & 791,475,447 & 14,444.27 & 5372 & 20 \\ \cline{3-10} 
 &  & WildsClosure (dir) & 41,570,951 & 41,570,951 & \textbf{50,000} & \textbf{41,071,218} & 7,028.63 & 5372 & 20 \\ \cline{3-10} 
 \cmidrule[2pt]{2-10} 
 &  \multirow{3}{*}{DG-Basis} & Closure (dir) & 56,730,368 & 113,460,737 & 114,550 & 155,929,003 & \textbf{6,626.68} & 2191 & 20 \\
\hline 
\end{tabular} 
} 
\caption{Total values of \textbf{\real datasets} per all analyzed measures: number of dependencies processed, number of operations on attributes, outer loops, inner loops  and computation time in miliseconds. $|\Sigma|$: size of the base. $|\mathcal{U}|$: number of attributes} 
\label{table:total_real_6} 
\end{table}

\begin{table} 
\centering 
\resizebox{\textwidth}{!}{ 
\begin{tabular}{|c||l|l||c|c||c||c|c|c|c|c|c|c|c|} 
\hline 
 DB & Base & Algorithm & deps & attrib & outer & inner & time (ms) & $|\Sigma|$ & $|\mathcal{U}|$ \\ 
\hline 
\hline 
\multirow{11}{*}{anneal} 
 &  \multirow{3}{*}{Canonical} & Closure (dir) & 9,485,079 & 18,970,158 & \textbf{50,000} & 63,250,000 & 1,225.89 & 1265 & 32 \\ \cline{3-10} 
 &  & Linclosure (dir) & 7,469,790 & 7,469,790 & 800,137 & 266,627,641 & 2,326.45 & 1265 & 32 \\ \cline{3-10} 
 &  & WildsClosure (dir) & 7,469,790 & 7,469,790 & \textbf{50,000} & 7,469,790 & 1,543.18 & 1265 & 32 \\ \cline{3-10} 
 \cmidrule[2pt]{2-10} 
 &  \multirow{3}{*}{D-Base} & Closure (dir) & 9,462,816 & 18,925,632 & \textbf{50,000} & 63,100,000 & 1,240.13 & 1262 & 32 \\ \cline{3-10} 
 &  & Linclosure (dir) & 8,204,621 & 8,204,621 & 820,756 & 266,014,644 & 2,365.48 & 1262 & 32 \\ \cline{3-10} 
 &  & WildsClosure (dir) & 8,204,621 & 8,204,621 & \textbf{50,000} & 7,404,484 & 1,564.35 & 1262 & 32 \\ \cline{3-10} 
 \cmidrule[2pt]{2-10} 
 &  \multirow{3}{*}{DG-Basis} & Closure (dir) & \textbf{6,637,984} & 13,275,967 & 100,722 & 70,180,586 & \textbf{1,141.70} & 791 & 32 \\ \cline{3-10} 
 &  & Linclosure & \textbf{6,637,984} & 19,913,951 & 863,221 & 235,402,303 & 2,608.20 & 791 & 32 \\ \cline{3-10} 
 &  & WildsClosure & \textbf{6,637,984} & \textbf{6,637,984} & 117,519 & \textbf{6,637,984} & 1,751.49 & 791 & 32 \\ \cline{3-10} 
\cmidrule[2pt]{1-10} 
\hline 
\end{tabular} 
} 
\caption{Total values of \textbf{\real datasets} per all analyzed measures: number of dependencies processed, number of operations on attributes, outer loops, inner loops  and computation time in miliseconds. $|\Sigma|$: size of the base. $|\mathcal{U}|$: number of attributes} 
\label{table:total_real_7} 
\end{table}


\section{Experiments with Synthetic Datasets}
\label{app:syntheticdatasets}


\begin{table} 
\centering 
\resizebox{\textwidth}{!}{ 

\begin{tabular}{|l|c||c|c|c|c||c|c|c|c||c|c|c|c|} 
\hline 
\multicolumn{2}{|c|}{\textbf{Processed Dependencies}} & \multicolumn{4}{|c|}{CDUB} & \multicolumn{4}{|c|}{D-Base} & \multicolumn{4}{|c|}{DG-basis} \\\hline 
 DB & $|\mathcal{U}|$ & $|\Sigma|$ & CLO & LIN & WILD & $|\Sigma|$ & CLO & LIN & WILD & $|\Sigma|$ & CLO & LIN & WILD \\ 
\hline 
\hline 
attrib-00005 & 5  & 8 & 1,328,800.95  & \textbf{974,788.68}  & \textbf{974,788.68}  & 5 & 1,024,588.03  & 1,349,194.45  & 1,349,194.45  & 5 & 1,152,573.94  & 1,152,573.94  & 1,152,573.94 \\ 
\hline 
attrib-00006 & 6  & 8 & 15,075,008.16  & 11,071,985.47  & 11,071,985.47  & 6 & \textbf{10,910,592.49}  & 13,420,502.02  & 13,420,502.02  & 6 & 11,617,168.68  & 11,617,168.68  & 11,617,168.68 \\ 
\hline 
attrib-00007 & 7  & 19 & 43,745,011.98  & \textbf{31,511,083.66}  & \textbf{31,511,083.66}  & 14 & 32,458,208.67  & 37,849,540.31  & 37,849,540.31  & 12 & 32,376,581.36  & 32,376,581.36  & 32,376,581.36 \\ 
\hline 
attrib-00008 & 8  & 18 & 104,953,938.12  & 72,633,708.32  & 72,633,708.32  & 16 & 75,041,785.05  & 80,326,400.68  & 80,326,400.68  & 12 & \textbf{72,549,609.05}  & \textbf{72,549,609.05}  & \textbf{72,549,609.05} \\ 
\hline 
attrib-00009 & 9  & 36 & 157,767,373.20  & 103,935,074.46  & 103,935,074.46  & 24 & 108,806,282.53  & 110,526,705.51  & 110,526,705.51  & 19 & \textbf{100,329,301.36}  & \textbf{100,329,301.36}  & \textbf{100,329,301.36} \\ 
\hline 
attrib-00010 & 10  & 40 & 239,672,712.43  & 145,643,222.82  & 145,643,222.82  & 23 & 164,222,758.27  & 150,498,387.84  & 150,498,387.84  & 17 & \textbf{136,172,863.65}  & \textbf{136,172,863.65}  & \textbf{136,172,863.65} \\
\hline 
\end{tabular} 
 
} 
\caption{Totals of the measure \textbf{Processed Dependencies} for all \textbf{\sintetic datasets}. In bold are the minimal values. The last line contains the average of each measure: the sum of all values for each pair (Base $\times$ Algorithm) divided by the number of datasets.} 
\label{table:total_deps_synthetic} 
\end{table}

\begin{table} 
\centering 
\resizebox{\textwidth}{!}{ 

\begin{tabular}{|l|c||c|c|c|c||c|c|c|c||c|c|c|c|} 
\hline 
\multicolumn{2}{|c|}{\textbf{Processed Attributes}} & \multicolumn{4}{|c|}{CDUB} & \multicolumn{4}{|c|}{D-Base} & \multicolumn{4}{|c|}{DG-basis} \\\hline 
 DB & $|\mathcal{U}|$ & $|\Sigma|$ & CLO & LIN & WILD & $|\Sigma|$ & CLO & LIN & WILD & $|\Sigma|$ & CLO & LIN & WILD \\ 
\hline 
\hline 
attrib-00005 & 5  & 8 & 2,657,601.90  & \textbf{974,788.68}  & \textbf{974,788.68}  & 5 & 2,049,176.06  & 1,349,194.45  & 1,349,194.45  & 5 & 2,305,147.88  & 3,457,721.82  & 1,152,573.94 \\ 
\hline 
attrib-00006 & 6  & 8 & 30,150,016.32  & \textbf{11,071,985.47}  & \textbf{11,071,985.47}  & 6 & 21,821,184.98  & 13,420,502.02  & 13,420,502.02  & 6 & 23,234,337.36  & 34,851,506.04  & 11,617,168.68 \\ 
\hline 
attrib-00007 & 7  & 19 & 87,490,023.96  & \textbf{31,511,083.66}  & \textbf{31,511,083.66}  & 14 & 64,916,417.34  & 37,849,540.31  & 37,849,540.31  & 12 & 64,753,162.72  & 97,129,744.08  & 32,376,581.36 \\ 
\hline 
attrib-00008 & 8  & 18 & 209,907,876.24  & 72,633,708.32  & 72,633,708.32  & 16 & 150,083,570.10  & 80,326,400.68  & 80,326,400.68  & 12 & 145,099,218.10  & 217,648,827.15  & \textbf{72,549,609.05} \\ 
\hline 
attrib-00009 & 9  & 36 & 315,534,746.40  & 103,935,074.46  & 103,935,074.46  & 24 & 217,612,565.06  & 110,526,705.51  & 110,526,705.51  & 19 & 200,658,602.72  & 300,987,904.08  & \textbf{100,329,301.36} \\ 
\hline 
attrib-00010 & 10  & 40 & 479,345,424.86  & 145,643,222.82  & 145,643,222.82  & 23 & 328,445,516.54  & 150,498,387.84  & 150,498,387.84  & 17 & 272,345,727.30  & 408,518,590.95  & \textbf{136,172,863.65} \\
\hline 
\end{tabular} 
 
} 
\caption{Totals of the measure \textbf{Processed Attributes} for all \textbf{\sintetic datasets}. In bold are the minimal values. The last line contains the average of each measure: the sum of all values for each pair (Base $\times$ Algorithm) divided by the number of datasets.} 
\label{table:total_atributs_synthetic} 
\end{table}

\begin{table} 
\centering 
\resizebox{\textwidth}{!}{ 

\begin{tabular}{|l|c||c|c|c|c||c|c|c|c||c|c|c|c|} 
\hline 
\multicolumn{2}{|c|}{\textbf{Running Time (milliseconds)}} & \multicolumn{4}{|c|}{CDUB} & \multicolumn{4}{|c|}{D-Base} & \multicolumn{4}{|c|}{DG-basis} \\\hline 
 DB & $|\mathcal{U}|$ & $|\Sigma|$ & CLO & LIN & WILD & $|\Sigma|$ & CLO & LIN & WILD & $|\Sigma|$ & CLO & LIN & WILD \\ 
\hline 
\hline 
attrib-00005 & 5  & 8 & 116.31  & 156.68  & 122.36  & 5 & \textbf{91.70}  & 134.95  & 114.98  & 5 & 108.99  & 311.57  & 168.04 \\ 
\hline 
attrib-00006 & 6  & 8 & 1,281.43  & 1,704.67  & 1,268.30  & 6 & \textbf{948.77}  & 1,325.07  & 1,068.96  & 6 & 1,094.59  & 3,150.00  & 1,681.84 \\ 
\hline 
attrib-00007 & 7  & 19 & 3,640.56  & 4,747.28  & 3,423.13  & 14 & \textbf{2,745.94}  & 3,740.97  & 2,879.34  & 12 & 3,039.08  & 8,685.11  & 4,529.17 \\ 
\hline 
attrib-00008 & 8  & 18 & 8,490.75  & 10,580.24  & 7,472.65  & 16 & 6,160.32  & 7,957.48  & \textbf{5,925.38}  & 12 & 6,707.12  & 19,093.34  & 9,681.32 \\ 
\hline 
attrib-00009 & 9  & 36 & 12,466.61  & 14,805.06  & 10,338.94  & 24 & 8,719.76  & 10,985.19  & \textbf{7,958.44}  & 19 & 9,133.36  & 26,021.14  & 12,891.52 \\ 
\hline 
attrib-00010 & 10  & 40 & 18,463.41  & 20,407.67  & 14,101.15  & 23 & 12,807.58  & 15,498.36  & \textbf{10,945.99}  & 17 & 12,160.28  & 34,612.01  & 16,416.22 \\
\hline 
\end{tabular} 
 
} 
\caption{Totals of the measure \textbf{Running Time (milliseconds)} for all \textbf{\sintetic datasets}. In bold are the minimal values. The last line contains the average of each measure: the sum of all values for each pair (Base $\times$ Algorithm) divided by the number of datasets.} 
\label{table:total_temps_synthetic} 
\end{table}

\begin{table} 
\centering 
\resizebox{\textwidth}{!}{ 
\begin{tabular}{|c||l|l||c|c||c||c|c|c|c|c|c|c|c|} 
\hline 
 DB & Base & Algorithm & deps & attrib & outer & inner & time (ms) & $|\Sigma|$ & $|\mathcal{U}|$ \\ 
\hline 
\hline 
\multirow{11}{*}{attrib-00010} 
 &  \multirow{3}{*}{Canonical} & Closure (dir) & 239,672,712 & 479,345,425 & \textbf{10,850,000} & 434,250,000 & 18,463.41 & 40 & 10 \\ \cline{3-10} 
 &  & Linclosure (dir) & 145,643,223 & 145,643,223 & 54,250,588 & 462,327,185 & 20,407.67 & 40 & 10 \\ \cline{3-10} 
 &  & WildsClosure (dir) & 145,643,223 & 145,643,223 & \textbf{10,850,000} & 145,643,223 & 14,101.15 & 40 & 10 \\ \cline{3-10} 
 \cmidrule[2pt]{2-10} 
 &  \multirow{3}{*}{D-Base} & Closure (dir) & 164,222,758 & 328,445,517 & \textbf{10,850,000} & 298,750,000 & 12,807.58 & 23 & 10 \\ \cline{3-10} 
 &  & Linclosure (dir) & 150,498,388 & 150,498,388 & 63,907,310 & 323,002,976 & 15,498.36 & 23 & 10 \\ \cline{3-10} 
 &  & WildsClosure (dir) & 150,498,388 & 150,498,388 & \textbf{10,850,000} & \textbf{96,247,799} & \textbf{10,945.99} & 23 & 10 \\ \cline{3-10} 
 \cmidrule[2pt]{2-10} 
 &  \multirow{3}{*}{DG-Basis} & Closure (dir) & \textbf{136,172,864} & 272,345,727 & 26,403,438 & 295,073,274 & 12,160.28 & 17 & 10 \\
\hline 
\end{tabular} 
} 
\caption{Total values of \textbf{\sintetic datasets} per all analyzed measures: number of dependencies processed, number of operations on attributes, outer loops, inner loops  and computation time in miliseconds. $|\Sigma|$: size of the base. $|\mathcal{U}|$: number of attributes} 
\label{table:total_synthetic_4} 
\end{table}

\begin{table} 
\centering 
\resizebox{\textwidth}{!}{ 
\begin{tabular}{|c||l|l||c|c||c||c|c|c|c|c|c|c|c|} 
\hline 
 DB & Base & Algorithm & deps & attrib & outer & inner & time (ms) & $|\Sigma|$ & $|\mathcal{U}|$ \\ 
\hline 
\hline 
\multirow{11}{*}{attrib-00013} 
 &  \multirow{3}{*}{Canonical} & Closure (dir) & 557,738,403 & 1,115,476,806 & \textbf{9,900,000} & 936,600,000 & 41,407.02 & 93 & 13 \\ \cline{3-10} 
 &  & Linclosure (dir) & 292,238,580 & 292,238,580 & 64,349,639 & 1,107,272,541 & 39,731.30 & 93 & 13 \\ \cline{3-10} 
 &  & WildsClosure (dir) & 292,238,580 & 292,238,580 & \textbf{9,900,000} & 292,238,580 & 27,205.04 & 93 & 13 \\ \cline{3-10} 
 \cmidrule[2pt]{2-10} 
 &  \multirow{3}{*}{D-Base} & Closure (dir) & 339,807,781 & 679,615,562 & \textbf{9,900,000} & 575,700,000 & 25,576.22 & 55 & 13 \\ \cline{3-10} 
 &  & Linclosure (dir) & 262,597,678 & 262,597,678 & 76,266,867 & 724,865,116 & 28,384.88 & 55 & 13 \\ \cline{3-10} 
 &  & WildsClosure (dir) & 262,597,678 & 262,597,678 & \textbf{9,900,000} & \textbf{198,248,039} & 19,533.79 & 55 & 13 \\ \cline{3-10} 
 \cmidrule[2pt]{2-10} 
 &  \multirow{3}{*}{DG-Basis} & Closure (dir) & \textbf{216,649,219} & 433,298,439 & 25,678,100 & 449,403,066 & \textbf{18,929.80} & 32 & 13 \\
\hline 
\end{tabular} 
} 
\caption{Total values of \textbf{\sintetic datasets} per all analyzed measures: number of dependencies processed, number of operations on attributes, outer loops, inner loops  and computation time in miliseconds. $|\Sigma|$: size of the base. $|\mathcal{U}|$: number of attributes} 
\label{table:total_synthetic_5} 
\end{table}

\begin{table} 
\centering 
\resizebox{\textwidth}{!}{ 
\begin{tabular}{|c||l|l||c|c||c||c|c|c|c|c|c|c|c|} 
\hline 
 DB & Base & Algorithm & deps & attrib & outer & inner & time (ms) & $|\Sigma|$ & $|\mathcal{U}|$ \\ 
\hline 
\hline 
\multirow{11}{*}{attrib-00019} 
 &  \multirow{3}{*}{Canonical} & Closure (dir) & 2,131,351,761 & 4,262,703,523 & \textbf{8,900,000} & 3,253,250,000 & 155,632.57 & 345 & 19 \\ \cline{3-10} 
 &  & Linclosure (dir) & 910,397,693 & 910,397,693 & 84,550,688 & 4,479,484,390 & 126,803.65 & 345 & 19 \\ \cline{3-10} 
 &  & WildsClosure (dir) & 910,397,693 & 910,397,693 & \textbf{8,900,000} & 910,397,693 & 87,320.96 & 345 & 19 \\ \cline{3-10} 
 \cmidrule[2pt]{2-10} 
 &  \multirow{3}{*}{D-Base} & Closure (dir) & 1,145,871,654 & 2,291,743,308 & \textbf{8,900,000} & 1,766,350,000 & 84,747.78 & 173 & 19 \\ \cline{3-10} 
 &  & Linclosure (dir) & 711,533,212 & 711,533,212 & 101,224,679 & 2,708,384,840 & 83,718.48 & 173 & 19 \\ \cline{3-10} 
 &  & WildsClosure (dir) & 711,533,212 & 711,533,212 & \textbf{8,900,000} & 626,982,524 & 57,447.26 & 173 & 19 \\ \cline{3-10} 
 \cmidrule[2pt]{2-10} 
 &  \multirow{3}{*}{DG-Basis} & Closure (dir) & \textbf{436,179,780} & 872,359,561 & 24,535,594 & 860,480,184 & \textbf{38,203.92} & 57 & 19 \\
\hline 
\end{tabular} 
} 
\caption{Total values of \textbf{\sintetic datasets} per all analyzed measures: number of dependencies processed, number of operations on attributes, outer loops, inner loops  and computation time in miliseconds. $|\Sigma|$: size of the base. $|\mathcal{U}|$: number of attributes} 
\label{table:total_synthetic_6} 
\end{table}

\begin{table} 
\centering 
\resizebox{\textwidth}{!}{ 
\begin{tabular}{|c||l|l||c|c||c||c|c|c|c|c|c|c|c|} 
\hline 
 DB & Base & Algorithm & deps & attrib & outer & inner & time (ms) & $|\Sigma|$ & $|\mathcal{U}|$ \\ 
\hline 
\hline 
\multirow{11}{*}{attrib-00023} 
 &  \multirow{3}{*}{Canonical} & Closure (dir) & 4,648,189,126 & 9,296,378,252 & \textbf{8,050,000} & 6,852,450,000 & 356,424.27 & 629 & 23 \\ \cline{3-10} 
 &  & Linclosure (dir) & 1,783,444,945 & 1,783,444,945 & 92,575,727 & 10,433,740,305 & 270,746.90 & 629 & 23 \\ \cline{3-10} 
 &  & WildsClosure (dir) & 1,783,444,945 & 1,783,444,945 & \textbf{8,050,000} & 1,783,444,945 & 189,373.62 & 629 & 23 \\ \cline{3-10} 
 \cmidrule[2pt]{2-10} 
 &  \multirow{3}{*}{D-Base} & Closure (dir) & 2,729,328,295 & 5,458,656,590 & \textbf{8,050,000} & 4,100,100,000 & 220,843.18 & 308 & 23 \\ \cline{3-10} 
 &  & Linclosure (dir) & 1,378,421,865 & 1,378,421,865 & 109,332,786 & 6,990,853,397 & 191,775.20 & 308 & 23 \\ \cline{3-10} 
 &  & WildsClosure (dir) & 1,378,421,865 & 1,378,421,865 & \textbf{8,050,000} & 1,285,846,137 & 132,060.25 & 308 & 23 \\ \cline{3-10} 
 \cmidrule[2pt]{2-10} 
 &  \multirow{3}{*}{DG-Basis} & Closure (dir) & \textbf{664,840,447} & 1,329,680,894 & 22,817,869 & 1,295,464,989 & \textbf{64,306.47} & 87 & 23 \\
\hline 
\end{tabular} 
} 
\caption{Total values of \textbf{\sintetic datasets} per all analyzed measures: number of dependencies processed, number of operations on attributes, outer loops, inner loops  and computation time in miliseconds. $|\Sigma|$: size of the base. $|\mathcal{U}|$: number of attributes} 
\label{table:total_synthetic_7} 
\end{table}

\begin{table} 
\centering 
\resizebox{\textwidth}{!}{ 
\begin{tabular}{|c||l|l||c|c||c||c|c|c|c|c|c|c|c|} 
\hline 
 DB & Base & Algorithm & deps & attrib & outer & inner & time (ms) & $|\Sigma|$ & $|\mathcal{U}|$ \\ 
\hline 
\hline 
\multirow{11}{*}{attrib-00029} 
 &  \multirow{3}{*}{Canonical} & Closure (dir) & 8,148,364,732 & 16,296,729,464 & \textbf{4,650,000} & 11,554,850,000 & 728,995.40 & 2055 & 29 \\ \cline{3-10} 
 &  & Linclosure (dir) & 2,791,030,259 & 2,791,030,259 & 67,424,966 & 19,380,699,366 & 536,211.67 & 2055 & 29 \\ \cline{3-10} 
 &  & WildsClosure (dir) & 2,791,030,259 & 2,791,030,259 & \textbf{4,650,000} & 2,791,030,259 & 341,872.09 & 2055 & 29 \\ \cline{3-10} 
 \cmidrule[2pt]{2-10} 
 &  \multirow{3}{*}{D-Base} & Closure (dir) & 4,870,034,141 & 9,740,068,282 & \textbf{4,650,000} & 7,091,650,000 & 475,498.05 & 1185 & 29 \\ \cline{3-10} 
 &  & Linclosure (dir) & 2,139,912,716 & 2,139,912,716 & 78,812,753 & 13,303,791,933 & 396,123.74 & 1185 & 29 \\ \cline{3-10} 
 &  & WildsClosure (dir) & 2,139,912,716 & 2,139,912,716 & \textbf{4,650,000} & 2,072,487,750 & 250,970.65 & 1185 & 29 \\ \cline{3-10} 
 \cmidrule[2pt]{2-10} 
 &  \multirow{3}{*}{DG-Basis} & Closure (dir) & \textbf{713,576,747} & 1,427,153,494 & 13,368,543 & 1,359,743,209 & \textbf{92,977.04} & 158 & 29 \\
\hline 
\end{tabular} 
} 
\caption{Total values of \textbf{\sintetic datasets} per all analyzed measures: number of dependencies processed, number of operations on attributes, outer loops, inner loops  and computation time in miliseconds. $|\Sigma|$: size of the base. $|\mathcal{U}|$: number of attributes} 
\label{table:total_synthetic_8} 
\end{table}

\begin{table} 
\centering 
\resizebox{\textwidth}{!}{ 
\begin{tabular}{|c||l|l||c|c||c||c|c|c|c|c|c|c|c|} 
\hline 
 DB & Base & Algorithm & deps & attrib & outer & inner & time (ms) & $|\Sigma|$ & $|\mathcal{U}|$ \\ 
\hline 
\hline 
\multirow{11}{*}{attrib-00035} 
 &  \multirow{3}{*}{Canonical} & Closure (dir) & 17,033,357,417 & 34,066,714,835 & \textbf{3,800,000} & 23,281,800,000 & 1,676,934.03 & 3435 & 35 \\ \cline{3-10} 
 &  & Linclosure (dir) & 5,325,944,102 & 5,325,944,102 & 66,498,781 & 41,692,133,559 & 1,203,232.40 & 3435 & 35 \\ \cline{3-10} 
 &  & WildsClosure (dir) & 5,325,944,102 & 5,325,944,102 & \textbf{3,800,000} & 5,325,944,102 & 685,860.06 & 3435 & 35 \\ \cline{3-10} 
 \cmidrule[2pt]{2-10} 
 &  \multirow{3}{*}{D-Base} & Closure (dir) & 10,622,101,279 & 21,244,202,558 & \textbf{3,800,000} & 14,761,600,000 & 1,139,936.12 & 1545 & 35 \\ \cline{3-10} 
 &  & Linclosure (dir) & 4,093,867,582 & 4,093,867,582 & 77,540,954 & 29,106,077,732 & 878,922.76 & 1545 & 35 \\ \cline{3-10} 
 &  & WildsClosure (dir) & 4,093,867,582 & 4,093,867,582 & \textbf{3,800,000} & 4,027,368,801 & 508,932.25 & 1545 & 35 \\ \cline{3-10} 
 \cmidrule[2pt]{2-10} 
 &  \multirow{3}{*}{DG-Basis} & Closure (dir) & \textbf{945,642,298} & 1,891,284,595 & 11,093,960 & 1,724,903,977 & \textbf{145,510.29} & 214 & 35 \\
\hline 
\end{tabular} 
} 
\caption{Total values of \textbf{\sintetic datasets} per all analyzed measures: number of dependencies processed, number of operations on attributes, outer loops, inner loops  and computation time in miliseconds. $|\Sigma|$: size of the base. $|\mathcal{U}|$: number of attributes} 
\label{table:total_synthetic_9} 
\end{table}

\begin{table} 
\centering 
\resizebox{\textwidth}{!}{ 
\begin{tabular}{|c||l|l||c|c||c||c|c|c|c|c|c|c|c|} 
\hline 
 DB & Base & Algorithm & deps & attrib & outer & inner & time (ms) & $|\Sigma|$ & $|\mathcal{U}|$ \\ 
\hline 
\hline 
\multirow{11}{*}{attrib-00038} 
 &  \multirow{3}{*}{Canonical} & Closure (dir) & 20,243,931,092 & 40,487,862,184 & \textbf{3,250,000} & 27,320,200,000 & 1,950,689.08 & 5895 & 38 \\ \cline{3-10} 
 &  & Linclosure (dir) & 6,181,718,876 & 6,181,718,876 & 61,750,701 & 49,602,125,976 & 1,372,411.14 & 5895 & 38 \\ \cline{3-10} 
 &  & WildsClosure (dir) & 6,181,718,876 & 6,181,718,876 & \textbf{3,250,000} & 6,181,718,876 & 778,238.50 & 5895 & 38 \\ \cline{3-10} 
 \cmidrule[2pt]{2-10} 
 &  \multirow{3}{*}{D-Base} & Closure (dir) & 12,469,426,529 & 24,938,853,058 & \textbf{3,250,000} & 17,125,400,000 & 1,296,028.15 & 2749 & 38 \\ \cline{3-10} 
 &  & Linclosure (dir) & 4,749,372,933 & 4,749,372,933 & 72,277,930 & 34,604,789,655 & 1,004,860.76 & 2749 & 38 \\ \cline{3-10} 
 &  & WildsClosure (dir) & 4,749,372,933 & 4,749,372,933 & \textbf{3,250,000} & 4,687,622,232 & 581,127.80 & 2749 & 38 \\ \cline{3-10} 
 \cmidrule[2pt]{2-10} 
 &  \multirow{3}{*}{DG-Basis} & Closure (dir) & \textbf{950,922,066} & 1,901,844,133 & 9,504,974 & 1,691,928,512 & \textbf{155,443.60} & 275 & 38 \\
\hline 
\end{tabular} 
} 
\caption{Total values of \textbf{\sintetic datasets} per all analyzed measures: number of dependencies processed, number of operations on attributes, outer loops, inner loops  and computation time in miliseconds. $|\Sigma|$: size of the base. $|\mathcal{U}|$: number of attributes} 
\label{table:total_synthetic_10} 
\end{table}

\begin{table} 
\centering 
\resizebox{\textwidth}{!}{ 
\begin{tabular}{|c||l|l||c|c||c||c|c|c|c|c|c|c|c|} 
\hline 
 DB & Base & Algorithm & deps & attrib & outer & inner & time (ms) & $|\Sigma|$ & $|\mathcal{U}|$ \\ 
\hline 
\hline 
\multirow{11}{*}{attrib-00043} 
 &  \multirow{3}{*}{Canonical} & Closure (dir) & 29,609,637,556 & 59,219,275,111 & \textbf{2,600,000} & 39,052,850,000 & 2,965,312.32 & 12773 & 43 \\ \cline{3-10} 
 &  & Linclosure (dir) & 8,634,490,716 & 8,634,490,716 & 55,897,464 & 72,775,217,509 & 2,089,653.06 & 12773 & 43 \\ \cline{3-10} 
 &  & WildsClosure (dir) & 8,634,490,716 & 8,634,490,716 & \textbf{2,600,000} & 8,634,490,716 & 1,113,936.41 & 12773 & 43 \\ \cline{3-10} 
 \cmidrule[2pt]{2-10} 
 &  \multirow{3}{*}{D-Base} & Closure (dir) & 17,666,248,140 & 35,332,496,280 & \textbf{2,600,000} & 23,742,600,000 & 2,047,326.44 & 10175 & 43 \\ \cline{3-10} 
 &  & Linclosure (dir) & 6,423,657,671 & 6,423,657,671 & 64,658,494 & 48,922,425,888 & 1,543,444.20 & 10175 & 43 \\ \cline{3-10} 
 &  & WildsClosure (dir) & 6,423,657,671 & 6,423,657,671 & \textbf{2,600,000} & 6,367,760,207 & 813,343.99 & 10175 & 43 \\ \cline{3-10} 
 \cmidrule[2pt]{2-10} 
 &  \multirow{3}{*}{DG-Basis} & Closure (dir) & \textbf{1,008,509,950} & 2,017,019,901 & 7,630,784 & 1,737,888,039 & \textbf{178,332.05} & 642 & 43 \\
\hline 
\end{tabular} 
} 
\caption{Total values of \textbf{\sintetic datasets} per all analyzed measures: number of dependencies processed, number of operations on attributes, outer loops, inner loops  and computation time in miliseconds. $|\Sigma|$: size of the base. $|\mathcal{U}|$: number of attributes} 
\label{table:total_synthetic_11} 
\end{table}

\begin{table} 
\centering 
\resizebox{\textwidth}{!}{ 
\begin{tabular}{|c||l|l||c|c||c||c|c|c|c|c|c|c|c|} 
\hline 
 DB & Base & Algorithm & deps & attrib & outer & inner & time (ms) & $|\Sigma|$ & $|\mathcal{U}|$ \\ 
\hline 
\hline 
\multirow{11}{*}{attrib-00051} 
 &  \multirow{3}{*}{Canonical} & Closure (dir) & 35,286,522,139 & 70,573,044,278 & \textbf{1,350,000} & 44,101,750,000 & 3,572,204.36 & 40367 & 51 \\ \cline{3-10} 
 &  & Linclosure (dir) & 9,469,942,090 & 9,469,942,090 & 34,426,935 & 85,133,176,050 & 2,448,857.89 & 40367 & 51 \\ \cline{3-10} 
 &  & WildsClosure (dir) & 9,469,942,090 & 9,469,942,090 & \textbf{1,350,000} & 9,469,942,090 & 1,228,248.93 & 40367 & 51 \\ \cline{3-10} 
 \cmidrule[2pt]{2-10} 
 &  \multirow{3}{*}{D-Base} & Closure (dir) & 18,854,996,746 & 37,709,993,493 & \textbf{1,350,000} & 23,751,050,000 & 2,129,195.14 & 23838 & 51 \\ \cline{3-10} 
 &  & Linclosure (dir) & 6,638,334,101 & 6,638,334,101 & 40,132,710 & 51,836,171,592 & 1,596,234.66 & 23838 & 51 \\ \cline{3-10} 
 &  & WildsClosure (dir) & 6,638,334,101 & 6,638,334,101 & \textbf{1,350,000} & 6,603,907,167 & 834,846.31 & 23838 & 51 \\ \cline{3-10} 
 \cmidrule[2pt]{2-10} 
 &  \multirow{3}{*}{DG-Basis} & Closure (dir) & \textbf{684,657,860} & 1,369,315,719 & 4,020,742 & 1,105,203,992 & \textbf{118,527.73} & 773 & 51 \\
\hline 
\end{tabular} 
} 
\caption{Total values of \textbf{\sintetic datasets} per all analyzed measures: number of dependencies processed, number of operations on attributes, outer loops, inner loops  and computation time in miliseconds. $|\Sigma|$: size of the base. $|\mathcal{U}|$: number of attributes} 
\label{table:total_synthetic_12} 
\end{table}

\begin{table} 
\centering 
\resizebox{\textwidth}{!}{ 
\begin{tabular}{|c||l|l||c|c||c||c|c|c|c|c|c|c|c|} 
\hline 
 DB & Base & Algorithm & deps & attrib & outer & inner & time (ms) & $|\Sigma|$ & $|\mathcal{U}|$ \\ 
\hline 
\hline 
\multirow{11}{*}{attrib-00057} 
 &  \multirow{3}{*}{Canonical} & Closure (dir) & 40,833,394,487 & 81,666,788,974 & \textbf{800,000} & 50,186,950,000 & 4,065,550.14 & 36911 & 57 \\ \cline{3-10} 
 &  & Linclosure (dir) & 10,433,588,779 & 10,433,588,779 & 22,799,082 & 100,885,489,375 & 2,929,868.72 & 36911 & 57 \\ \cline{3-10} 
 &  & WildsClosure (dir) & 10,433,588,779 & 10,433,588,779 & \textbf{800,000} & 10,433,588,779 & 1,400,318.52 & 36911 & 57 \\ \cline{3-10} 
 \cmidrule[2pt]{2-10} 
 &  \multirow{3}{*}{D-Base} & Closure (dir) & 22,142,275,243 & 44,284,550,485 & \textbf{800,000} & 27,268,400,000 & 2,537,498.54 & 19342 & 57 \\ \cline{3-10} 
 &  & Linclosure (dir) & 7,342,466,784 & 7,342,466,784 & 26,092,722 & 61,549,597,600 & 1,829,439.72 & 19342 & 57 \\ \cline{3-10} 
 &  & WildsClosure (dir) & 7,342,466,784 & 7,342,466,784 & \textbf{800,000} & 7,319,667,702 & 915,043.72 & 19342 & 57 \\ \cline{3-10} 
 \cmidrule[2pt]{2-10} 
 &  \multirow{3}{*}{DG-Basis} & Closure (dir) & \textbf{556,207,195} & 1,112,414,391 & 2,358,208 & 863,707,419 & 107,010.43 & 556 & 57 \\
\hline 
\end{tabular} 
} 
\caption{Total values of \textbf{\sintetic datasets} per all analyzed measures: number of dependencies processed, number of operations on attributes, outer loops, inner loops  and computation time in miliseconds. $|\Sigma|$: size of the base. $|\mathcal{U}|$: number of attributes} 
\label{table:total_synthetic_13} 
\end{table}

\begin{table} 
\centering 
\resizebox{\textwidth}{!}{ 
\begin{tabular}{|c||l|l||c|c||c||c|c|c|c|c|c|c|c|} 
\hline 
 DB & Base & Algorithm & deps & attrib & outer & inner & time (ms) & $|\Sigma|$ & $|\mathcal{U}|$ \\ 
\hline 
\hline 
\multirow{11}{*}{attrib-00068} 
 &  \multirow{3}{*}{Canonical} & Closure (dir) & 5,526,985,680 & 11,053,971,360 & \textbf{50,000} & 6,794,300,000 & 446,385.03 & 135886 & 68 \\ \cline{3-10} 
 &  & Linclosure (dir) & 1,399,554,898 & 1,399,554,898 & 1,700,079 & 13,707,253,376 & 346,938.74 & 135886 & 68 \\ \cline{3-10} 
 &  & WildsClosure (dir) & 1,399,554,898 & 1,399,554,898 & \textbf{50,000} & 1,399,554,898 & 169,159.31 & 135886 & 68 \\ \cline{3-10} 
 \cmidrule[2pt]{2-10} 
 &  \multirow{3}{*}{D-Base} & Closure (dir) & 3,252,257,370 & 6,504,514,739 & \textbf{50,000} & 3,945,600,000 & 347,640.91 & 78912 & 68 \\ \cline{3-10} 
 &  & Linclosure (dir) & 1,090,835,494 & 1,090,835,494 & 1,985,374 & 9,224,609,157 & 299,153.74 & 78912 & 68 \\ \cline{3-10} 
 &  & WildsClosure (dir) & 1,090,835,494 & 1,090,835,494 & \textbf{50,000} & 1,089,135,414 & 131,583.01 & 78912 & 68 \\ \cline{3-10} 
 \cmidrule[2pt]{2-10} 
 &  \multirow{3}{*}{DG-Basis} & Closure (dir) & \textbf{48,074,483} & 96,148,966 & 148,528 & 73,324,108 & 10,222.09 & 1085 & 68 \\
\hline 
\end{tabular} 
} 
\caption{Total values of \textbf{\sintetic datasets} per all analyzed measures: number of dependencies processed, number of operations on attributes, outer loops, inner loops  and computation time in miliseconds. $|\Sigma|$: size of the base. $|\mathcal{U}|$: number of attributes} 
\label{table:total_synthetic_14} 
\end{table}


\end{document}